\documentclass[12pt,a4paper,pdftex,english]{scrartcl}
\pdfoutput=1
\usepackage[latin1]{inputenc}
\usepackage[english]{babel}	
\usepackage{shadethm}
\usepackage{amsmath}
\usepackage{amsfonts}
\usepackage{amssymb}
\usepackage{amsthm}
\usepackage{algorithm}
\usepackage{algorithmic}
\usepackage{enumitem}
\usepackage{csquotes}
\usepackage{color}
\usepackage{graphics}
\usepackage{tikz}
\usepackage{url}
\usepackage{float}
\usetikzlibrary{arrows,decorations.pathreplacing,backgrounds,positioning,fit,shapes,calc}

\usepackage{enumitem}
\pgfkeys{%
/tikz/background rectangle/.style= {rounded corners,draw,fill=white},
/tikz/vertex/.style={circle,draw=black,fill=black, minimum size=1.5mm,inner sep=0pt},%
/tikz/vertex-out/.style={circle,draw=blue!50,fill=blue!5,minimum size=7mm},% 
/tikz/edge/.style={->,shorten <=1pt,>=stealth,thick},%
/tikz/undiredge/.style={-,>=stealth,thick},%
/tikz/edge-out/.style={->,shorten <=1pt,>=stealth,thick,densely dotted},%
/tikz/graphframe/.style={rectangle,draw,minimum height=5cm,minimum width=3cm},
/tikz/unlabvertex/.style={circle,draw=black,fill=black, minimum size=1.5mm,inner sep=0pt},%
/tikz/vector/.style={rectangle split,draw,rectangle split parts=3,rectangle split horizontal,rectangle split part fill={red!50, red!30, red!50},draw},%
/tikz/inner/.style={circle,draw=black,minimum size=7mm},
/tikz/sink/.style={rectangle,draw=black,minimum size=7mm},
/tikz/zero/.style={->,shorten <=1pt,>=stealth,semithick,densely dotted},
/tikz/one/.style={->,shorten <=1pt,>=stealth,semithick,solid},
/tikz/tnode/.style={circle,draw=black,inner sep=2pt},
/tikz/bits/.style={rectangle,draw=black,text height=10pt,text depth=3pt,inner sep=3pt,outer sep=0pt},
/tikz/entry/.style={rectangle,inner sep=2pt,minimum size=6mm}
}

\newcommand{\bset}{\lbrace 0,1 \rbrace}
\newcommand{\Prob}[2]{{\Pr\limits_{#1}\left[ #2 \right]}}
\newcommand{\iProb}[1]{{\Pr\left[ #1 \right]}}

\newcommand{\ie}{i.\,e., }
\newcommand{\iek}{i.\,e., }

\newcommand{\eg}{e.\,g., }

\definecolor{commentgreen}{rgb}{0.1,0.1,0.6}

\renewcommand{\algorithmicrequire}{\textbf{Input: }}
\renewcommand{\algorithmicensure}{\textbf{Output: }}

\theoremstyle{definition}
\newshadetheorem{definition}{Definition}[section]

\theoremstyle{plain}
\newshadetheorem{proposition}[definition]{Proposition}
\newshadetheorem{lemma}[definition]{Lemma}
\newshadetheorem{theorem}[definition]{Theorem}
\newshadetheorem{corollary}[definition]{Corollary}
\newtheorem{claim}[definition]{Claim}

\theoremstyle{remark}

\newcommand{\polylog}{\operatorname{polylog}}
\begin{document}

\title{OBDDs and (Almost) $k$-wise Independent Random Variables}

\author{Marc Bury\thanks{Supported by Deutsche Forschungsgemeinschaft, grant BO 2755/1-2.}\\TU Dortmund, LS2 Informatik, Germany}

\date{}

\maketitle 

\begin{abstract}
OBDD-based graph algorithms deal with the characteristic function of the edge set E of a graph $G = (V,E)$ which is represented by an OBDD and solve optimization problems by mainly using functional operations. We present an OBDD-based algorithm which uses randomization for the first time. In particular, we give a maximal matching algorithm with $O(\log^3 \vert V \vert)$ functional operations in expectation. This algorithm may be of independent interest. The experimental evaluation shows that this algorithm outperforms known OBDD-based algorithms for the maximal matching problem.

In order to use randomization, we investigate the OBDD complexity of $2^n$ (almost) $k$-wise independent binary random variables. We give a OBDD construction of size $O(n)$ for $3$-wise independent random variables and show a lower bound of $2^{\Omega(n)}$ on the OBDD size for $k \geq 4$. The best known lower bound was $\Omega(2^n/n)$ for $k \approx \log n$ due to Kabanets \cite{Kabanets03}. We also give a very simple construction of $2^n$ $(\varepsilon, k)$-wise independent binary random variables by constructing a random OBDD of width $O(n k^2/\varepsilon)$.
\end{abstract}

\newpage

\section{Introduction}
\label{sec:intro}

In times of Big Data, classical algorithms for optimization problems quickly exceed feasible running times or memory requirements. For instance, the rapid growth of the Internet and social networks results in massive graphs which traditional algorithms cannot process in reasonable time or space. In order to deal with such graphs, implicit (symbolic) algorithms have been investigated where the input graph is represented by the characteristic function $\chi_E$ of the edge set and the nodes are encoded by binary numbers. Using \emph{Ordered Binary Decision Diagrams} (OBDDs), which were introduced by Bryant \cite{Bryant86}, to represent $\chi_E$ can significantly decrease the space needed to store such graphs. Furthermore, using mainly functional operations, \eg binary synthesis and quantifications, which are efficiently supported by the OBDD data structure, many optimization problems can be solved on OBDD represented inputs (\cite{GentiliniPP03,GentiliniPP08,HachtelS97,Sawitzki04,SawitzkiSOFSEM06,SawitzkiLATIN06,Woe2006}). Implicit algorithms were successfully applied in many areas, \eg model checking \cite{BurchCMDH92}, integer linear programming \cite{LaiPV94} and logic minimization \cite{Coudert95}. With one of the first implicit graph algorithms, Hachtel and Somenzi \cite{HachtelS97} were able to compute a maximum flow on $0$-$1$-networks with up to $10^{36}$ edges and $10^{27}$ nodes in reasonable time.

There are two main parameters influencing the actual running time of OBDD-based algorithms: the number of functional operations and the sizes of all intermediate OBDDs used during the computation. The size of OBDDs representing graphs was investigated for bipartite graphs \cite{NuWo09}, interval graphs \cite{NuWo09,Gille13}, cographs \cite{NuWo09} and graphs with bounded tree- and clique-width \cite{MeerR09}. Bounding the sizes of the other OBDDs, which can occur during the computation, is quite difficult and could only be proven for very structured input graphs like grid graphs \cite{BolligGP12,Woe2006}. In terms of functional operations, Sawitzki \cite{Sawitzki07} showed that the set of problems solved by an implicit algorithm using $O(\log^k N)$ functional operations and functions defined on $O(\log N)$ variables is equal to the complexity class \textsc{FNC}, \iek the class of all optimization problems that can be efficiently solved in parallel. Implicit algorithms with these properties were designed for instance for topological sorting \cite{Woe2006},  minimum spanning tree \cite{Bollig12}, metric TSP approximation \cite{BolligC13} and maximal matching \cite{BolligP12} where a matching M, \ie a set of edges without a common vertex, is called maximal if M is no proper subset of another matching. However, Sawitzki's structural result yields neither a good transformation of parallel algorithms to implicit algorithms nor does it give a statement about the actual performance of the implicit algorithms. Nevertheless, designing implicit algorithms for optimization problems is not only an adaption of parallel algorithms but can give new insights into the problems. For example, Gentilini et al. \cite{GentiliniPP08} introduced a new notion of spine-sets in the context of implicit algorithms for connectivity related problems. When analyzing implicit algorithms, the actual running time can either be proven for very structured input graphs like \cite{Woe2006} did for topological sorting and \cite{BolligGP12} for maximum matching or the running time is experimentally evaluated like in \cite{HachtelS97} for maximum flows and in \cite{BolligGP12,Gille13} for maximum matching on bipartite graphs or unit interval graphs.

Overall there seems to be a trade-off: The number of operations is an important measure of difficulty \cite{BloemGS06} but decreasing the number of operations often results in an increase of the number of variables of the used functions. Since the worst case OBDD size of a function $f: \bset^n \rightarrow \bset$ is $\Theta(2^n/n)$, the number of variables should be as small as possible to decrease the worst case running time.
This trade-off was also empirically observed. For instance, an implicit algorithm computing the transitive closure that uses an iterative squaring approach and a polylogarithmic number of operations is often inferior to an implicit sequential algorithm, which needs a linear number of operations in worst case \cite{BloemGS06,HojatiTKB93}.
Another example is the maximal matching algorithm (BP) of Bollig and Pr\"{o}ger \cite{BolligP12} that uses only $O(\log^4 N)$ functional operations on functions with at most $6 \log N$ variables while the algorithm (HS) of Hachtel and Somenzi \cite{HachtelS97} uses $O(N \log N)$ operations in the worst case on function with at most $3 \log N$ variables. 
However, HS is clearly superior to BP on most instances (see Section \ref{sec:alg}). An additional reason might be its simplicity.

Using randomization in an explicit algorithm often leads to simple and fast algorithms. Here, we propose the first attempt at using randomization to obtain algorithms which have both a small number of variables and a small expected number of functional operations. For this, we want to represent random functions $f_r: \bset^n \rightarrow \bset$ with $\iProb{f_r(x) = 1} = p$ for every $x \in \bset^n$ and some fixed probability $0 < p < 1$ by OBDDs. Using random functions in implicit algorithms is difficult. We need to construct them efficiently but, obviously, if the function values are completely independent (and $p$ is a constant), then the OBDD (and even the more general FBDD or read-once branching program) size of $f_r$ is exponentially large with an overwhelming probability \cite{Weg94a}. Thus, we investigate the OBDD size and construction of (almost) $k$-wise independent random functions where the distribution induced on every $k$ different function values is (almost) uniform. 

\subsubsection*{Related Work} 
A succinct representation of $2^n$ random bits, which are $k$-wise independent, was presented by Alon et al. \cite{Alon86} using $\lfloor k/2 \rfloor n+1$ independent random bits. 
This number of random bits is very close to the lower bound of Chor and Goldreich \cite{ChorG89}. In order to reduce the number of random bits even further, Naor and Naor \cite{NaorN93} introduced the notion of almost $k$-wise independence where the distribution on every $k$ random bits is \enquote{close} to uniform. Constructions of almost $k$-wise independent random variables are also given in \cite{AlonGHP92} and are using only at most $2(\log n + \log k + \log(1/\varepsilon))$ random bits where $\varepsilon$ is a bound on the closeness to the uniform distribution.
Looking for a simple representation of almost $k$-wise independent random variables, Savický \cite{Savicky95} presented a Boolean formula of constant depth and polynomial size and used $n\log^2 k \log(1/\varepsilon)$ random bits. In all of these constructions, the running time of computing the $i$-th random bit with $0 \leq i \leq 2^n-1$ depends on $k$ and $\varepsilon$.

Such small probability spaces can be used for a succinct representation of a random string of length $2^n$, \eg in streaming algorithms \cite{AlonMS99}, or for derandomization \cite{Alon86,Luby86}. The randomized parallel algorithms from \cite{Alon86,Luby86} compute a maximal independent set (MIS) of a graph, \ie a subset $I$ of $V$ such that no two nodes of $I$ are adjacent and any vertex in $G$ is either in $I$ or is adjacent to a node of $I$. The computation of a MIS has also been extensively studied in the area of distributed algorithms \cite{AwerbuchGLP89,Linial92}. An optimal randomized distributed MIS algorithm was presented in \cite{MetivierRSZ11} where the time and bit complexity (bits per channel) is $O(\log N)$. Using completely independent random bits, Israeli and Itai \cite{IsraelI86} give a randomized parallel algorithm computing a maximal matching in time $O(\log N)$. 

While we are looking for $k$-wise independent functions with small OBDD size, Kabanets \cite{Kabanets03} constructed simple Boolean functions which are hard for FBDDs by investigating (almost) $\Theta(n)$-wise independent random functions and showed that the probability tends to $1$ as $n$ grows that the size is $\Omega(2^n/n)$. 

\subsubsection*{Our Contribution.} In Section \ref{sec:compl}, we show that the OBDD and FBDD size is at least $2^{\Omega(n + \log(p'))}$ with $p' = 2p(1-p)$ if the function values of $f_r$ are $k$-wise independent with $k \geq 4$. We give an efficient construction of OBDDs for $3$-wise independent random functions which is based on the known construction of $3$-wise independent random variables using BCH-schemes \cite{Alon86}. In Section \ref{sec:construction} we investigate a simple construction of a random OBDD due to Bollig and Wegener \cite{BolligW14} which generates almost $k$-wise independent random functions and has size $O((kn)^2/\varepsilon)$. Reading the actual value of the $i$-th random bit is just an evaluation of the function on input $i$ which can be done in $O(n)$ time, \ie it is independent of both $k$ and $\varepsilon$. This construction is used as an input distribution for our implicit algorithm in the experimental evaluation. In Section \ref{sec:alg} we use pairwise independent random functions to design a simple maximal matching algorithm that uses only $O(\log^3 N)$ functional operations in expectation and functions with at most $3 \log N$ variables. This algorithm can easily be extended to the MIS problem and can be implemented as a parallel algorithm using $O(\log N)$ time in expectation or as a distributed algorithm with $O(\log N)$ expected time and bit complexity (and is simpler than in \cite{MetivierRSZ11}). To the best of our knowledge, this is the first (explicit or implicit) maximal matching (or independent set) algorithm that does not need any knowledge about the graph (like size or node degrees) as well as uses only pairwise independent random variables. Eventually, we evaluate this algorithm empirically and show that known implicit maximal matching algorithms are outperformed by the new randomized algorithm.

\section{Preliminaries}
\label{sec:pre}

\subsubsection*{Binary Decision Diagrams}
\label{sec:pre_obdd}
We denote the set of Boolean functions $f: \bset^n \rightarrow \bset$ by $B_n$. For $x \in \bset^n$ denote the value of $x$ by $\vert x \vert := \sum_{i=0}^{n-1} x_i \cdot 2^i$. Further, for $l \in \mathbb{N}$, we denote by $[l]_2$ the corresponding binary number of $l$, \ie $\vert [l]_2 \vert = l$. In his seminal paper \cite{Bryant86}, Bryant introduced \emph{Ordered Binary Decision Diagrams} (OBDDs), that allow a compact representation of not too few Boolean functions and also supports many functional operations efficiently.

\begin{definition}[Ordered Binary Decision Diagram (OBDD)]\ 

\textbf{Order.} A \emph{variable order} $\pi$ on the input variables $X = \lbrace x_0, \ldots, x_{n-1} \rbrace$ of a Boolean function $f \in B_n$ is a permutation of the index set $I = \lbrace 0, \ldots, n-1 \rbrace$.

\textbf{Representation.} A \emph{$\pi$-OBDD} is a directed, acyclic, and rooted graph $G$ with two sinks labeled by the constants $0$ and $1$. Each inner node is labeled by an input variable from $X$ and has exactly two outgoing edges labeled by $0$ and $1$. Each edge $(x_i,x_j)$ has to respect the variable order $\pi$, \ie $\pi(i) < \pi(j)$. 
%The $\pi$-OBDD is called \emph{complete} if every variable is tested on every path.

\textbf{Evaluation.} An assignment $a \in \bset^n$ of the variables defines a path from the root to a sink by leaving each $x_i$-node via the $a_i$-edge. A $\pi$-OBDD $G_f$ represents $f$ iff for every $a \in \bset^n$ the defined path ends in the sink with label \nolinebreak $f(a)$.

\textbf{Complexity.} The \emph{size} of a $\pi$-OBDD $G$, denoted by $\vert G \vert$, is the number of nodes in $G$. The $\pi$-OBDD size of a function $f$ is the minimum size of a $\pi$-OBDD representing $f$. The OBDD size of $f$ is the minimum $\pi$-OBDD size over all variable orders $\pi$. 
%The minimal-size complete OBDD for $f$ is called \emph{quasi-reduced} OBDD.
The \emph{width} of $G$ is the maximum number of nodes labeled by the same input variable.
\end{definition}

The more general read-once branching programs or \emph{Free Binary Decision Diagrams} (FBDDs) were introduced by Masek \cite{Masek76}. In an FBDD every variable can only be read once on a path from the root to a sink (but the order is not restricted).

A simple function is the inner product $IP_n(x,y) = \bigoplus_{i=0}^{n-1} x_i \wedge y_i$ of two vectors $x,y \in \bset^n$. Let $\pi$ be a variable order where for every $0 \leq i \leq n$, the variables $x_i$ and $y_i$ are consecutive. 
%W.l.o.g. let $\pi = (x_0,y_0, \ldots, x_{n-1},y_{n-1})$. Define $b_j := \bigoplus_{i=0}^j x_i \wedge y_i$ for $0\leq j \leq n-1$. If we know the value of $b_j$ for some $j$, we can easily extend it to the value of $b_{j+1}$ by reading the pair $(x_{j+1}, y_{j+1})$ since $b_{j+1} = b_j \oplus (x_{j+1} \wedge y_{j+1})$. 
It is easy to see that the $\pi$-OBDD representing $IP_n$ has size $O(n)$ and width $2$. Notice that the $\pi$-OBDD size is still $O(n)$ if we replace an input vector, \eg $y$, by a constant vector $r \in \bset^n$.

In the following we describe some important operations on Boolean functions which we will use in this paper  (see, \eg Section 3.3 in \cite{Wegener00} for a detailed list). 
Let $f$ and $g$ be Boolean functions in $B_n$ on the variable set 
$X=\{x_0, \ldots, x_{n-1}\}$, $\pi$ a fixed order and let $G_f$ and $G_g$ be $\pi$-OBDDs representing $f$ and $g$, respectively. We denote the subfunction of $f$ where $x_j$ for some $0 \leq j \leq n-1$ is replaced by a constant $a \in \bset$ by $f_{\mid x_j=a}$.

\begin{enumerate}
\item \textbf{Negation:} Given $G_f$, compute a representation for the function $\overline{f} \in B_n$. Time: $O(1)$
\item \textbf{Replacement by constant:} Given $G_f$, an index $i\in \{0, \ldots, n-1\}$, and a Boolean constant $c_i\in \{0,1\}$, compute a representation for the subfunction $f_{|x_i=c_i}$. Time: $O(\vert G_f \vert)$
\item \textbf{Equality test:} Given $G_f$ and $G_g$, decide whether $f$ and $g$ are equal. Time: $O(1)$ in most implementations (when using so called \emph{Shared OBDDs}, see \cite{Wegener00}), otherwise $O(\vert G_f \vert + \vert G_g \vert)$.
%\item \textbf{Satisfiability:} Given $G_f$, decide whether $f$ is not the constant function $0$. Time: $O(1)$
%\item \textbf{Satisfiability count:} Given $G_f$, compute $|f^{-1}(1)|$. Time: $O(\vert G_f \vert)$
\item \textbf{Synthesis:} Given $G_f$ and $G_g$ and a binary Boolean operation $\otimes \in B_2$, compute a representation for the function $h\in B_n$ defined as $h:=f\otimes g$. Time: $O(\vert G_f \vert \cdot \vert G_g \vert)$
%or $O(\vert G_h \vert \cdot \log \vert G_h \vert)$
\item \textbf{Quantification:} Given $G_f$, an index $i\in \{1, \ldots, n\}$ and a quantifier $Q\in\{\exists, \forall\}$, compute a representation for the function $h\in B_n$ defined as $h:=Q x_i: f$ where  $\exists x_i: f := f_{|x_i=0}\vee f_{|x_i=1}$ and $\forall x_i: f := f_{|x_i=0}\wedge f_{|x_i=1}$. Time: see replacement by constant and synthesis
\end{enumerate}

In addition to the operations mentioned above, in implicit graph algorithms (see the next section) the following operation (see, \eg \cite{SawitzkiLATIN06}) is useful to reverse the edges of a given graph. We will use this operation implicitly by writing for instance $f(x,y)$ and $f(y,x)$ in the pseudo code of our algorithm.
\begin{definition}
Let $k \in \mathbb{N}$, $\rho$ be a permutation of $\{1,\ldots, k\}$ and
$f\in B_{kn}$ with input vectors $x^{(1)}, \ldots, x^{(k)} \in \bset^n$. The argument reordering $\mathcal{R}_\rho(f)\in B_{kn}$ with respect to $\rho$ is defined by $\mathcal{R}_\rho(f)(x^{(1)}, \ldots, x^{(k)}) := f(x^{(\rho (1))}, \ldots, x^{(\rho (k))})$.
\end{definition}
This operation can be computed by just renaming the variables and repairing the variable order using $3(k-1)n$ functional operations (see \cite{BolligLW96}).

 A function $f$ \emph{depends essentially} on a variable $x_i$ iff $f_{\mid x_i = 0} \neq f_{\mid x_i = 1}$. A characterization of minimal $\pi$-OBDDs due to Sieling and Wegener \cite{SielingW93} can often be used to bound the OBDD size. 
\begin{theorem}[\cite{SielingW93}]
\label{thm:minimal_obdd}
Let $f \in B_n$ and for all $i=0,\ldots,n-1$ let $s_i$ be the number of different subfunctions which result from replacing all variables $x_{{\pi(j)}}$ with $0 \leq j \leq i-1$ by constants and which essentially depend on $x_{\pi(i)}$. Then the minimal $\pi$-OBDD representing $f$ has $s_i$ nodes labeled by $x_{\pi(i)}$. 
\end{theorem}

Lower bound techniques for FBDDs are similar but have to take into account that the order can change for different paths. The following property due to Jukna \cite{Jukna88} can be used to show good lower bounds for the FBDD size.

\begin{definition}
A function $f \in B_n$ with input variables $X=\{x_0, \ldots, x_{n-1}\}$ is called \emph{$r$-mixed} if for all $V \subseteq X$ with $\vert V \vert = r$ the $2^r$ assignments to the variables in $V$ lead to different \nolinebreak subfunctions.
\end{definition}

\begin{lemma}[\cite{Jukna88}]
\label{lem:kmixed}
The FBDD size of a $r$-mixed function is bounded below by $2^r-1$.
\end{lemma}

\subsubsection*{OBDD-Based Graph Algorithms} Let $G = (V,E)$ be a directed graph with node set $V = \lbrace v_0, \ldots, v_{N-1} \rbrace$ and edge set $E \subseteq V \times V$. Here, an undirected graph is interpreted as a directed symmetric graph. Implicit algorithms work on the characteristic function $\chi_E \in B_{2n}$ of $E$ where $n = \lceil \log N \rceil$ is the number of bits needed to encode a node of $V$ and $\chi_E(x,y) = 1$ if and only if $(v_{\vert x \vert}, v_{\vert y \vert}) \in E$. Often it is also necessary to store the valid encodings of nodes by the characteristic function $\chi_V$ of $V$. Besides functional operations, OBDD-based algorithms can use $O(\polylog \vert V \vert)$ additional time, \eg for constructing OBDDs for a specific function (equality, greater than, inner product, ...).

\subsubsection*{Small probability spaces} A succinct representation of our random function is essential for our randomized implicit algorithm. For this, we have to used random functions with limited independence.
\begin{definition}[(Almost) $k$-wise independence]
Let $X_0,\ldots,X_{m-1}$ be $m$ binary random variables. These variables are called \emph{$k$-wise independent} with $k \leq m$ if and only if for all $0 \leq i_1 < \ldots i_k \leq m-1$ and for all $l_1,\ldots,l_k \in \bset$
\begin{equation*}
 \Prob{}{X_{i_1}=l_1 \wedge \ldots \wedge X_{i_k} = l_k}  =  2^{-k} 
\end{equation*}
and they are called \emph{$(\varepsilon,k)$-wise independent} iff
$$ \vert \Prob{}{X_{i_1}=l_1 \wedge \ldots \wedge X_{i_k} = l_k} - 2^{-k} \vert \leq \varepsilon. $$
\end{definition}
The BCH scheme introduced by Alon et. al \cite{Alon86} is a construction of $k$-wise independent random variables $X_0,\ldots,X_{2^n-1}$ that only needs $\lfloor k/2 \rfloor n+1$ independent random bits and works as follows: Let $r_n \in \bset$ be a random bit, $r^{(j)} \in \bset^n$ for $1 \leq j \leq l$ be $l$ uniformly random row vectors, and let the row vector $r = \left[ r^{(1)}, \ldots, r^{(l)} \right] \in \bset^{ln+1}$ be the concatenation of the vectors. For $0 \leq i \leq 2^n-1$ define
$ X_i = IP_{ln+1}\left(r, \left[ \left[ i \right]_2, \left[ i^3 \right]_2, \ldots, \left[ i^{2l-1} \right]_2 \right] \right) \oplus r_n $
where $i^{2j-1}$ for $j=1,\ldots,l$ is computed in the finite field $GF(2^{n})$. This scheme generates $2l+1$-wise independent random bits \cite{Alon86} (if we exclude $X_0$ and if $r_n$ is dropped we obtain $2l$-wise independence). 

%The number of random bits used in BCH schemes is very close to the lower bound of Chlor et al. \cite{Chlor}. Naor and Naor introduced the notion of almost $k$-wise independent random variables to be able to further reduce the number of random bits.

%\begin{definition}[Almost $k$-wise Independence]
%Let $X_1,\ldots,X_m$ be $m$ binary random variables. These variables are called \emph{$(\varepsilon,k)$-wise independent} with $k \leq m$ if and only if for all $1 \leq i_1 < \ldots i_k \leq m$ and for all $l_1,\ldots,l_k \in \bset$
%$$ \vert \Prob{}{X_{i_1}=l_1 \wedge \ldots \wedge X_{i_k} = l_k} - 2^{-k} \vert \leq \varepsilon. $$
%\end{definition}

We say a function $f_r: \bset^n \rightarrow \bset$ is a $k$-wise ($(\varepsilon,k)$-wise) independent random function iff the random variables $X_i = f_r(\left[ i \right]_2)$ with $0 \leq i \leq 2^n-1$ are $k$-wise ($(\varepsilon,k)$-wise) independent. The BCH scheme gives us an intuition of the complexity of an OBDD representing a $k$-wise independent function: For $k \leq 3$ the random variables of the BCH scheme are
$X_i = IP_{ln+1}\left(r, \left[ i \right]_2 \right) \oplus r_n$
which is basically a simple inner product of two binary vectors. For $k \geq 4$, \ie $l \geq 2$, we have to multiply in a finite field to generate the random variables. Since multiplication is hard for OBDDs it seems likely that $k$-wise independent functions for $k \geq 4$ are also hard.

%Finally, we need some tail bounds for $k$-wise independent random variables to prove our lower bounds in the next section. Let $X = \sum_{i=0}^N X_i$ be a sum of $k$-wise independent variables with $k \geq 2$. It is easy to verify that $Var\left[ X \right] = \sum_{i=0}^N Var\left[ X_i \right]$. Therefore, if we can calculate or bound the variance of every $X_i$ we can use Chebyshev\'{}s inequality:
%$$ \Prob{}{\vert X - E\left[ X \right] \vert \geq \delta} \leq \dfrac{Var\left[ X \right]}{\delta^2} = \dfrac{\sum_{i=0}^N Var\left[ X_i \right]}{\delta^2}. $$

\section{OBDD Size of $k$-wise Independent Random Functions}
\label{sec:compl}

%\subsection*{OBBD Size of $k$-wise Independent Random Functions}
We start with some upper bounds on the OBDD size of $3$-wise independent random functions. Notice that by means of the BCH scheme it is not possible to construct a pairwise independent function (which is not $3$-wise independent) since $X_0 = IP(r,0^n) = 0$ for every $r \in \bset^n$.
\begin{theorem}
\label{th:obdd_construction}
Let $\varepsilon > 0$, $n \in \mathbb{N}$, $p$ be a probability with $1/2^n \leq p \leq 1/2$, and $\pi$ be a variable order on the input variables $\lbrace x_0, \ldots, x_{n-1} \rbrace$. Define $p(x) := \Prob{r\in \bset^{n+1}}{f_r(x) = 1}$.
\begin{enumerate}
\item We can construct an $\pi$-OBDD representing a $3$-wise independent function $f_r: \bset^n\rightarrow \bset$ in time $O(n)$ such that $p(x) = 1/2 $ for every $x \in \bset^n$, and the size of the $\pi$-OBDD is $O(n)$ with width $2$ for every $r \in \bset^{n+1}$ (see Algorithm \ref{alg:implrandfunc}).
\item We can construct an $\pi$-OBDD representing a $3$-wise independent function $f_r: \bset^n \rightarrow \bset$ in time $O(\frac{n}{p \cdot \varepsilon})$ such that $ \frac{\lceil p \cdot 2^n \rceil}{2^n} \leq p(x) \leq (1+\varepsilon) \cdot \frac{\lceil p \cdot 2^n \rceil}{2^n}$ for every $x \in \bset^n$, and the size of the $\pi$-OBDD is bounded above by $O(\frac{n}{p \cdot \varepsilon})$ for every $r\in \bset^{n+1}$.
\item We can compute a function $g_A(x): \bset^n \leftarrow \bset$ for a random matrix $A \in \bset^{n} \times \bset^n$ such that $\Prob{A}{g_A(x) = 1} = \dfrac{\lceil p \cdot 2^n \rceil}{2^n}$ using $O(n)$ functional operations. Furthermore, using also $O(n)$ functional operations we can compute a priority function $GT_A(x,y)$, which is equal to $1$ iff $\vert v_x \vert > \vert v_y \vert$, where $v_z \in \bset^n$ for all $z \in \bset^n$ and $(v_0, \ldots, v_{2^n-1})$ is a pairwise independent random permutation of $\bset^n$. Note: This construction is also possible using only $2n$ random bits by computing single bits of $a_0+ \vert x \vert \cdot a_1$ in $\mathbb{F}_{2^n}$.
\end{enumerate}
\end{theorem}

\begin{figure}
\begin{center}
\begin{tikzpicture}
	\node (x00)   at (2,4)   [inner] {$x_0$};
	\node (x01-1) at (1,3.2)   [inner] {$x_2$}; 
	\node (x01-2) at (3,3.2)   [inner] {$x_2$}; 
	\node (x02-1) at (1,2)   [inner] {$x_3$}; 
	\node (x02-2) at (3,2)   [inner] {$x_3$}; 
	\node (x03-1) at (1,0.8)   [inner] {$x_5$}; 
	\node (x03-2) at (3,0.8) 	[inner] {$x_5$}; 

	\node (00)    at (1.3,-0.3)[sink]  {$0$};
	\node (01)    at (2.7,-0.3)[sink]  {$1$};

	\draw [zero] (x00) to (x01-1);
	\draw [one]  (x00) to (x01-2);
	
	\draw [zero] (x01-1) to (x02-1);
	\draw [one] (x01-1) to (x02-2);
	
	\draw [zero] (x01-2) to (x02-2);
	\draw [one] (x01-2) to (x02-1);	

	\draw [zero] (x02-1) to (x03-1);
	\draw [one] (x02-1) to (x03-2);
	
	\draw [zero] (x02-2) to (x03-2);
	\draw [one] (x02-2) to (x03-1);
	
	\draw [zero] (x03-1) to (00);
	\draw [one] (x03-1) to (01);
	
	\draw [zero] (x03-2) to (01);
	\draw [one] (x03-2) to (00);	
	
	\node (x0)   at (6,4.2)   [inner] {$x_0$};
	\node (y0) at (6.2,3)   [inner] {$y_0$}; 
	\node (x1-1) at (5,2)   [inner] {$x_1$}; 
	\node (x1-2) at (7,2)   [inner] {$x_1$}; 
	\node (y1-1) at (5,0.8)   [inner] {$y_1$}; 
	\node (y1-2) at (7,0.8) 	[inner] {$y_1$}; 

	\node (0)    at (5.3,-0.3)[sink]  {$0$};
	\node (1)    at (6.7,-0.3)[sink]  {$1$};

	\draw [zero] (x0) to (x1-1);
	\draw [one]  (x0) to (y0);
	
	\draw [zero] (y0) to (x1-1);
	\draw [one]  (y0) to (x1-2);
	
	\draw [zero,bend left] (x1-1) to (0);
	\draw [one]  (x1-1) to (y1-1);	

	\draw [zero,bend right] (x1-2) to (1);
	\draw [one]  (x1-2) to (y1-2);
	
	\draw [zero] (y1-1) to (0);
	\draw [one]  (y1-1) to (1);
	
	\draw [zero] (y1-2) to (1);
	\draw [one]  (y1-2) to (0);
	
\end{tikzpicture}
\end{center}
\caption{Two $\pi$-OBBDs with $\pi = (x_0,y_0, \ldots, x_{n-1},y_{n-1})$ for the functions $IP_6(x,y)$ where $y$ is replaced by the constant vector $(1,0,1,1,0,1)$ and $IP_2(x,y)$.}\label{fig:obdd_ip}
\end{figure}
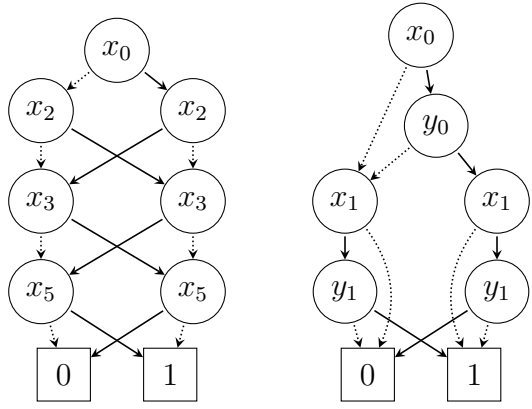

\begin{proof}
1. This is an implication of the BCH scheme for $3$-wise independent random bits. Recall that for random $r = (r_0, \ldots, r_{n}) \in \bset^{n+1}$ the random variables $X_i(r) = IP(r, \left[ 1,\left[ i \right]_2 \right])$ for $0 \leq i \leq 2^n-1$ are $3$-wise independent and $\Prob{}{X_i = 1} = 1/2$ for every $i$. We define $f_r(x) = X_{\vert x \vert}(r)$ (see Algorithm \ref{alg:implrandfunc}). As described in the preliminaries, the function $IP$ can be represented by an OBDD of width $2$ and size $O(n)$ for any variable order if one input vector is replaced by a constant vector. The construction of the OBDD is straightforward (see, \eg Fig. \ref{fig:obdd_ip}) and can be done in time $O(n)$.

2. Let $s \in \bset^n$ be the binary representation of $\lceil p \cdot 2^n \rceil$, \ie $\vert s \vert = \lceil p \cdot 2^n \rceil$. In order to approximate the probability $p$, we compute $t = \lceil - \log p - \log \varepsilon \rceil$ random bits $c_{n-1}(x),\ldots,c_{n-t}(x)$ with $c_i(x) = IP(x,r^{(i)}) \oplus r^{(i)}_n$ where the $r^{(i)} \in \bset^n$ and $r^{(i)}_n \in \bset$ are chosen independently uniformly at random. Now, our random function $f_r(x)$ is equal to $1$ iff $\vert c_{n-1} \cdots c_{n-t} 0^{n-t} \vert \leq \vert s \vert$. In order to construct the OBDD for $f_r$, we simulate the $t$ OBDDs representing $c_i$ on input $x \in \bset^n$ in parallel: Since all OBDDs have width $2$, we can represent the states of the OBDDs representing $c_{n-1}, \ldots, c_{n-t}$ after reading the same $k$ input variables with at most $2^t$ OBDD nodes for each $1 \leq k \leq n$. After reading all $n$ input bits, we know the values of $c_{n-1}(x), \ldots, c_{n-t}(x)$ and can easily decide whether $\vert c_{n-1} \cdots c_{n-t} 0^{n-t} \vert \leq \vert s \vert$ because $s$ is a constant. Therefore, the overall OBDD size is $O(n \cdot 2^t) = O(\frac{n}{p \cdot \varepsilon})$.
%Thus, we define
%$$ f_r(x) = \bigwedge\limits_{i=n-t}^{n-1} (s_i \vee \overline{c_i(x)}). $$
%Since $s$ is a constant and the OBDD for every $c_i$ has size $O(n)$ and width $2$, the resulting OBDD size is $O(n \cdot 2^t) = O(\frac{n}{p \cdot \varepsilon})$ (see Lemma \ref{lem:widthsynthesis}). 
Each $c_i(x)$ is generated by a BCH scheme for $3$-wise independent random bits and $c_i(x)$ and $c_j(x)$ are independent for $i \neq j$. Therefore, $f_r(x)$ is also $3$-wise independent.

Let $p' \in (0,1)$ such that $2^n \cdot p'$ is the value of the binary number consisting of the first $t$ most significant bits of $s$ followed by $n-t$ ones, \ie  $2^n \cdot p' = \vert s_{n-1} \cdots s_{n-t} 1^{n-t} \vert$. The function $f_r$ ignores the $n-t$ least significant bits of $s$, therefore, it is equivalent to choose a random vector $v \in \bset^n$ and check whether $\vert v \vert \leq 2^n \cdot p'$ which means that the probability of $f_r(x) = 1$ is $p'$. It is
$$\vert s \vert \leq p' \cdot 2^n \leq \vert s \vert + 2^{n-t}-1 \leq \vert s \vert + \varepsilon \cdot p \cdot 2^{n} \leq (1+\varepsilon) \vert s \vert $$
and thus 
$$ \dfrac{\lceil p \cdot 2^n \rceil}{2^n} \leq \Prob{}{f_r(x) = 1} = p' \leq (1+\varepsilon) \cdot \dfrac{\lceil p \cdot 2^n \rceil}{2^n}. $$
3. For a random matrix $A \in \bset^n \times \bset^n$ define $r_i(x) = a_i^T \cdot x = IP(a_i,x)$ where $a_i$ is the $i$-th column vector of $A$. Let $f_A(x,i) = 1$ iff $r_i(x) = 1$. An OBDD reading the bits of $i$ first and then computing the corresponding value of $r_i(x)$ has a size of $O(n^2)$ since each $r_i(x)$ can be computed by an OBDD with width $2$ (TODO: experiments suggest rather size of $O(n)$). This OBDD can also be constructed in time $O(n^2)$. Let $v_x = Ax$ for $x \in \bset^n$. Now define $GT_A(x,y) = 1$ iff $\vert v_x \vert > \vert v_y \vert$, \ie $$ GT_A(x,y) = \exists i: f_A(x,i) \wedge \overline{f_A(y,i)} \wedge (\forall j: (j > i) \Rightarrow (f_A(x,i) \Leftrightarrow f_A(y,i))).$$
Since each component of $v_x$ and $v_y$ are pairwise independent, the entire random vectors are also pairwise independent. 

The function $g_A(x)$ can be constructed in the same way by replacing $f_A(y,i)$ with $b(i)$ which is equal to the $i$-th bit of $\lceil p \cdot 2^n \rceil$ and a negation of the resulting function.
\end{proof}

\begin{algorithm}[t]
\caption{RandomFunc(x,n)}\label{alg:implrandfunc}
\algorithmicrequire{Variable vector $x$ of length $n \in \mathbb{N}$}\\
\algorithmicensure{$3$-wise independent function $r(x)$}
\begin{algorithmic}
\STATE Let $r_0, \ldots, r_{n}$ be $n+1$ independent random bits
\STATE $f_r(x) = \bigoplus_{i=0}^{n-1} (r_i \wedge x_i) \oplus r_n$
\RETURN $f_r(x)$
\end{algorithmic}
\end{algorithm}

Can we also construct small OBDDs for $k$-wise independent random variables with $k \geq 4$? Unfortunately, this is not possible.

\begin{theorem}
\label{th:lower_bound_fixedorder}
Let $X_0,\ldots,X_{2^n-1} : S \rightarrow \bset$ be $k$-wise independent $0/1$-random variables over a sample space $S$ with $\Prob{}{X_j = 1} = p$ for all $0 \leq j \leq 2^n-1$ and $k \geq 4$. For every $s \in S$ let $f_s: \bset^n \rightarrow \bset$ be defined by $f_s(x) := X_{\vert x \vert}(s)$. Then, for a fixed variable order $\pi$, the expected $\pi$-OBDD size of $f_s$ is bounded below by $\Omega(2^{n/3} \cdot (p')^{(1/3)})$ with $p' = 2p(1-p)$.
\end{theorem}
\begin{proof}
For the sake of simplicity we omit the index of the function $f_r$. W.l.o.g. let $\pi$ be the identity order, i.e. $\pi(i) = i$ for all $i=0,\ldots,n-1$. For $l \in \lbrace 1,\ldots,n \rbrace$ and $\alpha \in \bset^l$ let $f_\alpha: \bset^{n-l} \rightarrow \bset$ be the subfunction of $f$ where the first $l$ variables $x_0,\ldots,x_{l-1}$ are fixed according to $\alpha$, i.e. $f_\alpha(z) := f_{\mid x_0=\alpha_0,\ldots,x_{l-1}=\alpha_{l-1}}(z)$. Now we fix two different assignments $\alpha,\alpha'$ and define $2^{n-l}$ random variables $D(z) := D_{\alpha,\alpha'}(z)$ such that $D(z) = 1$ iff $f_{\alpha}(z) \neq f_{\alpha'}(z)$. Since the function values of $f$ are also $k$-wise independent, for every $z \in \bset^{n-l}$ we have 
$ E\left[ D(z) \right] = 2p(1-p) := p'$ and $ Var\left[ D(z) \right] = E\left[ D(z)^2 \right] - E\left[ D(z) \right]^2 = E\left[ D(z) \right] - E\left[ D(z) \right]^2 = p'(1-p')$.
Let $D = \sum_z D(z)$. We want to find an upper bound on the number of pairs $(\alpha,\alpha')$ with $f_{\alpha} = f_{\alpha'}$. The probability that for fixed $(\alpha,\alpha')$ the subfunctions are equal is bounded above by the probability that the difference between $D$ and $E\left[D\right]$ is at least $E\left[D\right]$, i.e.
$ \Prob{}{f_\alpha = f_{\alpha'}} = \Prob{}{D = 0} \leq \Prob{}{\vert D - E\left[D\right] \vert \geq E\left[D\right]}$.
Each random variable $D(z)$ depends on two function values, i.e. these variables are $k' = \lfloor k/2 \rfloor$-wise independent. Since $k' \geq 2$ we can use Chebyshev\'{}s inequality
$$ \Prob{}{f_\alpha = f_{\alpha'}} \leq \dfrac{Var\left[D\right]}{E\left[D\right]^2} = \dfrac{\sum_z Var\left[D(z)\right]}{(2^{n-l}\cdot p')^2} = \dfrac{2^{n-l}\cdot p' \cdot (1-p')}{(2^{n-l}\cdot p')^2} \leq \dfrac{1}{2^{n-l} \cdot p'} $$
Hence, the expected number of pairs $(\alpha,\alpha')$ with $f_\alpha = f_\alpha'$ is bounded above by 
$ \frac{\binom{2^{l}}{2}}{2^{n-l} \cdot p'} \leq \frac{2^{2l}}{2^{n-l} \cdot p'}$.
Therefore, the expected number $t_l$ of functions which are equal can be bounded above by 
$ \sqrt{\frac{2^{2l}}{2^{n-l} \cdot p'}} = \frac{2^l}{\sqrt{2^{n-l} \cdot p'}} $.
The number $s_l$ of different subfunctions $f_\alpha$ is bounded below by $2^l$ divided by an upper bound on the number $t_l$ of subfunctions $f_\alpha$ which are equal, \ie
$ E\left[ s_l \right] \geq E\left[ \frac{2^l}{t_l} \right] \geq \frac{2^l}{E\left[ t_l \right]} $
where the last inequality is due to Jensen's inequality and the fact that $g(x) = x^{-1}$ is convex on $(0,\infty)$.
The expected number of equal subfunctions can be lower than $1$, therefore we have to do a case study:
\begin{enumerate}
\item $ \dfrac{2^l}{\sqrt{2^{n-l} \cdot p'}} \leq 1 \Leftrightarrow l \leq (1/3)(n+\log(p'))$: All $2^l$ subfunctions are different (in expectation).
%\begin{eqnarray*}
% \dfrac{2^l}{\sqrt{2^{n-l} \cdot p'}} \leq 1 & \Leftrightarrow & 2^l \leq \sqrt{2^{n-l+\log(p')}}\\
% & \Leftrightarrow & l \leq (1/2)(n-l+\log(p'))\\
% & \Leftrightarrow & l \leq (1/3)(n+\log(p'))
%\end{eqnarray*}
\item $\dfrac{2^l}{\sqrt{2^{n-l} \cdot p'}} > 1$: The number of different subfunctions is at least 
$ \frac{2^l \cdot \sqrt{2^{n-l} \cdot p'}}{2^l} = \sqrt{2^{n-l} \cdot p'}$.
\end{enumerate}
For the sake of simplicity, we assume that $(1/3)(n+\log(p'))$ is an integer, since this does not affect the asymptotic behavior. 
Due to the first case, we know that the number of different subfunctions has to double after each input bit on the first $(n/3)+\log(p')/3+1$ levels, \ie each node must have two outgoing edges to two different nodes which also means that the all subfunctions essentially depend on the next variable. Therefore, the $\pi$-OBDD has to be a complete binary tree on the first $(n/3)+\log(p')/3+1$ levels and the expected $\pi$-OBDD size is also $\Omega(2^{n/3}\cdot (p')^{1/3})$.

%Summing up, we get an expected lower bound of 
%$$ \sum\limits_{l=0}^{ (n/3)+\log(p')/3 } 2^l + \sum\limits_{l= (n/3)+\log(p')/3 +1}^{n} \sqrt{2^{n-l} \cdot p'} = \Omega(2^{n/3}\cdot (p')^{1/3})
%$$
% \\
% = & 2^{n/3+1} \cdot (p')^{1/3}-1 + \sqrt{p'} \cdot \sum\limits_{i=0}^{(2/3)n-\log(p')/3-1} \sqrt{2^i}\\
%= & 2^{n/3+1} \cdot (p')^{1/3}-1 +  \sqrt{p'} \cdot \dfrac{1-\sqrt{2}^{(2/3)n-\log(p')/3-1}}{1-\sqrt{2}}\\
%= & 2^{n/3+1} \cdot (p')^{1/3}-1 +  \sqrt{p'} \cdot \dfrac{2^{(n/3)-\log(p')/6-1/2}-1}{\sqrt{2}-1}\\
%= & \Omega(2^{n/3}\cdot (p')^{1/3})+\Omega(2^{n/3} \cdot (p')^{1/3}) = \Omega(2^{n/3}\cdot (p')^{1/3})
%\end{array}$$
%on the size of the quasi-reduced $\pi$-OBDD (since in case 2 we do not ensure that the subfunctions have to be essentially dependent on the next variable). But we know (due to the first case) that the number of different subfunctions has to double after each input bit on the first $(n/3)+\log(p')/3+1$ levels, \ie each node must have two outgoing edges to two different nodes which also means that the all subfunctions essentially depend on the next variable. Therefore, the $\pi$-OBDD has to be a complete binary tree on the first $(n/3)+\log(p')/3+1$ levels and the expected $\pi$-OBDD size is also $\Omega(2^{n/3}\cdot (p')^{1/3})$.
\end{proof}
%The last theorem states that we cannot use $k$-wise independent random functions in an implicit algorithm with a fixed variable order and $k \geq 4$. But it is still possible that for every function $f_s(x)$ there is a variable order $\pi$ such that the $\pi$-OBDD representing $f_s$ is small. 
The following theorem shows that $k$-wise independent random functions with $k \geq 4$ are hard even for FBDDs (and with it for OBDDs and all variable orders).
%Wegener \cite{Weg94a} showed that completely independent random functions have a large OBDD size with overwhelming probability. Here, we show that even for $k$-wise independent random functions with $k \geq 4$ there are functions with large OBDD size.
The general strategy of the proof of the next theorem is similar to the proof in \cite{Weg94a} where the OBDD size of completely independent random functions was analyzed: We bound the probability $p_l$ that there is a variable order such that the number of OBDD nodes on level $l$ deviates too much from the expected value. If $\sum_{l = 0}^{n-1} p_l < 1$ holds, then with probability $1-\sum_{l = 0}^{n-1} p_l > 0$ there is no such deviation in any level of the OBDD for all variable orders. The differences lie in the detail: In \cite{Weg94a} the function values are completely independent and, therefore, the calculation can be done more directly and with better estimations. We have to take the detour over the number of subfunctions which are equal (as in Theorem \ref{th:lower_bound_fixedorder}) and can use only Markov's inequality to calculate the deviation of the expectation. Furthermore, because of the independence Wegener \cite{Weg94a} was able to do a more subtle analysis of the OBDD size by investigating the effects of the OBDD minimization rules separately.

\begin{theorem}
\label{th:lower_bound_general}
Let $X_0,\ldots,X_{2^n-1}: S \rightarrow \bset$ be $k$-wise independent $0/1$-random variables over a sample space $S$ with $\Prob{}{X_j = 1} = p$ for all $0 \leq j \leq 2^n-1$ and $k \geq 4$. For every $s \in S$ let $f_s: \bset^n \rightarrow \bset$ be defined by $f_s(x) := X_{\vert x \vert}(s)$. Then, there is an $r$-mixed function $f_s$ with $r = \Omega(n + \log(p')-\log n)$ with $p' = 2p(1-p)$.
%such that the OBDD size is bounded below by $2^{\Omega(n + \log(p')-\log n)}$ with $p' = 2p(1-p)$.
\end{theorem}
\begin{proof} 
%We use a similar idea to prove this theorem as in \cite{Weg94a}: We want to bound the probability $p_i$ that there is a variable order such that the number of OBDD nodes on level $i$ deviates from the expected value. We have to ensure that $\sum p_i < 1$ and the deviation is not too large to complete the proof. 
%The two differences to the proof in \cite{Weg94a} are the ways to calculate both the number of OBDD nodes on level $i$ (see Theorem \ref{th:lower_bound_fixedorder}) and the deviation from the expectation. 
First, we bound the probability that the number $t_l$ of subfunctions which are equal deviates by a factor of $\delta_l$ from the expectation. Second, as in Theorem \ref{th:lower_bound_fixedorder}, we show an upper bound on the level $l$ for which the number of equal subfunctions is lower or equal than $1$, \ie the OBDD has to be a complete binary tree until this level.

As we know, the expected number of pairs $(\alpha,\alpha')$ with $f_\alpha = f_\alpha'$ is bounded above by $\mu_l := \frac{2^{2l}}{2^{n-l} \cdot p'}$. Due to the dependencies, using Markov's inequality is the best we can do to bound the deviation from the expectation. Thus, we have
$$ \Prob{}{\text{No. pairs } (\alpha,\alpha')\text{ with }f_\alpha = f_\alpha' \geq \delta_l \cdot \mu_l} \leq \frac{1}{\delta_l}.$$
Due to Theorem \ref{thm:minimal_obdd}, the definition of the subfunctions corresponding to OBDD nodes on level $l$, \ie the definition of the subfunctions $f_\alpha$, depends only on the first $l$ variables with respect to the variable order. Thus, we have to distinguish only $\binom{n}{l}$ possibilities to choose these variables. Let $\delta_l := \binom{n}{l} \cdot (n+1)$. Then the probability, that for all levels and variable orders the number of pairs $(\alpha,\alpha')$ with $f_\alpha = f_{\alpha'}$ is at most $\delta_l \cdot \mu_l$ is bounded below by $1-n/(n+1) > 0$. Note that this also implies that $t_l \leq \sqrt{\delta_l \cdot \mu_l}$ for all levels $l$ and variable orders.

As in the proof of Theorem \ref{th:lower_bound_fixedorder}, the next step consists of the investigation of two cases: $\sqrt{\delta_l \mu_l} \leq 1$ and $\sqrt{\delta_l \mu_l} > 1$. Here, we focus only on the first case. In other words, we compute an upper bound $T$ such that $\sqrt{\delta_l \mu_l} \leq 1$ or, equivalently, $(1/2) \log(\delta_l \mu_l) \leq 0$ for all $l \leq T$. For the calculations, we need a known bound for the binomial coefficient $ \log\binom{n}{k} \leq n \cdot H(k/n)$
where $H(x) = -x \log(x)-(1-x)\log(1-x)$ is the binary entropy function. It holds
\begin{eqnarray*}
%\frac{1}{2}\log(\delta_l \mu_l) & = &  \frac{1}{2}\log\left(\binom{n}{l} \cdot (n+1) \cdot \frac{2^{2l}}{2^{n-l} \cdot p'}\right)\\
\frac{1}{2}\log(\delta_l \mu_l) & \leq & \frac{1}{2}(3l-n+\log(n)+1-\log(p')+ n \cdot H\left(\frac{l}{n}\right)).
%& \leq & \frac{1}{2}(\log\binom{n}{l} + \log(2n) + 2l - (n-l) - \log(p'))\\
%& \leq & \frac{1}{2}(3l-n+\log(n)+1-\log(p')+ n \cdot H\left(\frac{l}{n}\right)).
\end{eqnarray*}
Let $l=\varepsilon \cdot n$ for some $\varepsilon < 1/2$. We want to maximize $\varepsilon$ such that $\log(\sqrt{\delta_l \mu_l}) \leq 0$.
$$\begin{array}{lrcl}
& \dfrac{1}{2}(3(\varepsilon n)-n+\log(n)+1-\log(p')+n \cdot H(\varepsilon)) & \leq &  0\\
\Leftrightarrow & 3\varepsilon+H(\varepsilon) & \leq & 1-\dfrac{\log(1/p')}{n}-\dfrac{1}{n}-\dfrac{\log n}{n}
\end{array}$$
%\Leftrightarrow & (3\varepsilon+H(\varepsilon))n+1+\log(n) & \leq & n+\log(p')\\
Using $\dfrac{1}{1-x} \leq 1+2x$ for $0 \leq x \leq 1/2$ and $\log(1+x) \leq x$ for $x > -1$, we can bound $3\varepsilon+H(\varepsilon)$ by $6 \sqrt{\varepsilon}$ (see Appendix for the details).
%\begin{eqnarray*}
%3\varepsilon+H(\varepsilon) & = & 3\varepsilon + \varepsilon  \log(1/\varepsilon)+(1-\varepsilon) \cdot \log(1/(1-\varepsilon))\\
%& \leq & 3\varepsilon + \varepsilon \sqrt{1/\varepsilon} + \log(1/(1-\varepsilon))\\
%& \leq & 3\varepsilon+\sqrt{\varepsilon} + 2\varepsilon\\
%& \leq & 6 \sqrt{\varepsilon}. 
%\end{eqnarray*}
Thus, if 
$$ \varepsilon \leq \sqrt{\varepsilon} \leq \dfrac{1}{6}-\dfrac{1}{6} \cdot \left(\dfrac{\log(1/p')}{n}+\dfrac{1}{n}+\dfrac{\log n}{n}\right) = \Omega\left(1-\frac{\log(1/p')+\log n}{n}\right), $$
it is $\log(\sqrt{\delta_l \mu_l}) \leq 0$.
Since $l=\varepsilon \cdot n$ and the maximal $\varepsilon$ is in $\Omega\left(1-\frac{\log(1/p')+\log n}{n}\right)$ such that $\log(\sqrt{\delta_l \mu_l}) \leq 0$, there is a function $f_s$ which is $r$-mixed with $r = \Omega(n + \log(p')-\log n)$.
\end{proof}
Due to Lemma \ref{lem:kmixed}, the last Theorem gives us an lower bound even for FBDDs.
%Recalling the proof of Theorem \ref{th:lower_bound_general}, it is easy to see that even the FBDD size has to be exponentially large for some function $f_s$. 
\begin{corollary}
Let $X_0,\ldots,X_{2^n-1}$ be $k$-wise independent $0/1$-random variables over a sample space $S$ with $\Prob{}{X_j = 1} = p$ for all $0 \leq j \leq 2^n-1$ and $k \geq 4$. For every $s \in S$ let $f_s: \bset^n \rightarrow \bset$ be defined by $f_s(x) := X_{\vert x \vert}(s)$. Then, there is a function $f_s$ such that the FBDD size is at least $2^{\Omega(n+\log(p')-\log n)}$.
\end{corollary}
%\begin{proof}
%We know that for large $n \in \mathbb{N}$ there is a function $f_s$ and a level $k = \Omega(n+\log(p')-\log n)$ such that the expected number of nodes in this level is $2^k$ for every variable order. This implies that for every set of $k$ variables all possible assignments lead to different subfunctions, \ie the function is $k$-mixed. Due to Lemma \ref{lem:kmixed}, the FBDD size is at least $2^k-1 = 2^{\Omega(n+\log(p')-\log n)}$.
%\end{proof}

\section{Construction of Almost $k$-wise Independent Random Functions.}
\label{sec:construction}
The gap between the OBDD size of $3$-wise independent random functions and $4$-wise independent random functions is exponentially large. In order to see what kind of random functions have an OBDD size which is in between these bounds, we show that a construction of a random OBDD due to \cite{BolligW14} of size $O((nk)^2/\varepsilon)$ generates $(\varepsilon, k)$-wise independent functions. The idea is to construct a random OBDD with fixed width $w$. If $w$ is large enough, the function values of $k$ different inputs are almost uniformly distributed because the paths of the $k$ inputs in the OBDD are likely to be almost independent. For $0 \leq i \leq n-1$ let layer $L_i$ consists of $w$ nodes labeled by $x_i$ and layer $L_n$ be the two sinks. For all $0 \leq i \leq n-1$ we choose the $0$/$1$-successors of every node in layer $L_i$ independently and uniformly at random from the nodes in layer $L_{i+1}$. Then we pick a random node in layer $L_0$ as the root of the OBDD.

\begin{theorem}
For $w \geq k+nk(k+1)/\varepsilon$ the above random process generates $(\varepsilon, k)$-wise independent random functions.
\end{theorem}
\begin{proof}
Let $a_1, \ldots, a_k \in \bset^n$ be $k$ different inputs and $p$ be the probability that the function values of these inputs are $\alpha_1,\ldots,\alpha_k \in \bset$. Let $P_1, \ldots, P_k$ the $k$ paths of $a_1,\ldots,a_k$ to the layer $L_{n-1}$, \ie the paths end in a node labeled by $x_{n-1}$. Let $D_i$ be the event that the paths $P_1,\ldots,P_i$ end in different nodes. Since the inputs are different, every $P_i$ has to use an edge which is not used by any other path and, therefore, it holds $\Prob{}{D_i \mid D_{i-1}} \geq (1-\frac{i-1}{w})^n$ and with it $\Prob{}{D_k} = \prod\limits_{i=2}^{k} \Prob{}{D_i \mid D_{i-1}} \geq \prod\limits_{i=2}^{k} (1-\frac{i-1}{w})^n$. We have 
$ \prod\limits_{i=2}^{k} (1-\frac{i-1}{w})^n \geq \prod\limits_{i=2}^{k} e^{-\frac{n}{w/i-1}} \geq 1-\varepsilon $ for $w \geq k+nk(k+1)/\varepsilon \geq k+nk(k+1)(1/\ln(\frac{1}{1-\varepsilon}))$. If all paths end in different nodes, then the function values of the $k$ inputs are independent and uniformly distributed, \ie $p \geq 2^{-k} \cdot \Prob{}{D_k} \geq 2^{-k}-\varepsilon$ and $p \leq 1-(1-2^{-k})\cdot \Prob{}{D_k} \leq 2^{-k}+\varepsilon$ which completes the proof.
\end{proof}

%If we are allowed to do repeated test of variables, we have to investigate the communication complexity of the function (TODO). 

%\paragraph{TODO.} More randomized algorithms :) Both with a polylogarithmic number of operations (in expectation) and heuristics. Bonus: Number of operations not polylogarithmic in deterministic setting and here with polylogarithmic operations in expectation! Idea for new random priority function: Let $r \in \bset^n$ be a random vector, then define
%$$ d_r(x,y) = 1 \Leftrightarrow \vert x - r \vert \leq \vert y - r \vert. $$
%General idea: Let $b(p)$ the binary number of $\lceil p \cdot 2^n \rceil$. Simulate the random number for each $x \in \bset^n$ bit by bit and check whether the number is greater than $b(p)$. More general, $b(x): \bset^n \rightarrow \mathbb{N}$ and let $b_i(x) $ a function that maps $x$ to the $i$-th bit of a $b(x)$ (note that $b_i$ needs to be efficiently computable). Then to the simulation as above. 

\section{Randomized Implicit Algorithms}
\label{sec:alg}

\subsubsection*{Complexity Class} Only a small modification is necessary to extend Sawitzki's simulation results from \cite{SawitzkiSOFSEM06} and \cite{Sawitzki07} to show that the set of problems which can solved by a randomized implicit algorithm is equal to the set of problems solved by a randomized parallel algorithm. In the implicit setting, we just add the possibility to construct random functions $r: \bset^l \rightarrow \bset$ with $l = O(\log N)$. Constructing such functions in parallel is easy. The other way round, \ie simulating a randomized parallel algorithm by a randomized implicit algorithm, the only difference is the set of input variables of the circuit (which represents the (randomized) parallel algorithm). A deterministic circuit has only $N$ input variables whereas the random circuit has additional $O(N^c)$ random inputs for a constant $c$. Assuming we can construct a random function $r: \bset^l \rightarrow \bset$ with $l = O(\log N)$, we can set the input variables correctly for the simulation (in the same way as in \cite{Sawitzki07}).

\begin{algorithm}[h]
\caption{Randomized implicit maximal matching algorithm}
\label{alg:implrandmatch}
\algorithmicrequire{Graph $\chi_E(x,y)$}\\
\algorithmicensure{Maximal matching $\chi_M(x,y)$}
\begin{algorithmic}
\STATE $\chi_M(x,y) = 0$ \hfill \COMMENT{Initial matching}
\WHILE{$\chi_E(x,y) \not\equiv 0$}
\STATE $\chi_{E'}(x,y) = \chi_E(x,y)$
\STATE \COMMENT{Compute set of nodes with two or more incident edges}
\STATE $T(x) = \exists z,y: (z \neq y) \wedge \chi_{E'}(x,y) \wedge \chi_{E'}(x,z)$
\STATE $NewEdges(x,y) = 0$
\WHILE{$T(x) \not\equiv 0$}
\STATE \COMMENT{Construct $3$-wise independent random functions (see Algorithm \ref{alg:implrandfunc})}
\STATE $f_{r_1}(x) = RandomFunc(x,n)$ and $f_{r_2}(y) = RandomFunc(y,n)$ 
\STATE $F(x,y) = (x > y) \wedge (f_{r_1}(x) \oplus f_{r_2}(y))$
\STATE $F(x,y) = F(x,y) \vee F(y,x)$
\STATE $\chi_{E'}(x,y) = \chi_{E'}(x,y) \wedge F(x,y)$ \hfill \COMMENT{Delete edges with probability $1/2$}
\STATE $T(x) = \exists z,y: (z \neq y) \wedge \chi_{E'}(x,y) \wedge \chi_{E'}(x,z)$ \hfill \COMMENT{Update $T(x,y)$}
\STATE \COMMENT{Store isolated edges in NewEdges}
\STATE $NewEdges(x,y) = NewEdges(x,y) \vee (\chi_{E'}(x,y) \wedge \overline{T(x)} \wedge \overline{T(y)})$
\ENDWHILE
\STATE $\chi_M(x,y) = \chi_M(x,y) \vee NewEdges(x,y)$ \hfill \COMMENT{Add edges to current matching}
\STATE $Matched(x) = \exists y: \chi_M(x,y)$ 
\STATE $\chi_E(x,y) = \chi_E(x,y) \wedge \overline{Matched(x)} \wedge \overline{Matched(y)}$ \hfill \COMMENT{Delete edges incident to matched nodes}
\ENDWHILE
\RETURN $\chi_M(x,y)$
\end{algorithmic}
\end{algorithm}

\subsubsection*{Randomized Maximal Matching Algorithm}
We use the construction of $3$-wise independent random functions from the last section to design a randomized maximal matching algorithm. Here, the main drawback of our random construction is the missing possibility to use different probabilities for the nodes. Randomized algorithms for maximal independent set using pairwise independence like in \cite{Alon86} or \cite{Luby86} choose a node with a probability proportional to the node degree. In order to simulate these selections by our construction, we delete each edge with probability $1/2$ as long as there are other incident edges. Finally, we add the remaining isolated edges to the matching.
Algorithm \ref{alg:implrandmatch} shows the whole randomized implicit maximal matching algorithm. We realize the edge deletions of the inner loop in the following way: We construct two $3$-wise independent random functions $f_{r_1}(x), f_{r_2}(y)$ using Algorithm \ref{alg:implrandfunc} and set $ F(x,y) = (x > y) \wedge (f_{r_1}(x) \oplus f_{r_2}(y))$.
Since  $\Prob{r_1,r_2}{f_{r_1}(x) \oplus f_{r_2}(y) = 1} = \Prob{r_1,r_2}{f_{r_1}(x) \neq f_{r_2}(y)} = 1/4+1/4 = 1/2$ for inputs $x \leq y$ the function $F(x,y)$ deletes such edges as required. Because we are dealing with undirected graphs, we want $F(x,y) = F(y,x)$ for every $(x,y)$. Therefore, we set $F(x,y) = F(x,y) \vee F(y,x)$ and delete the edges with the operation $\chi_E(x,y) = \chi_E(x,y) \wedge F(x,y)$.

We say that an edge $e \in E'$ (before the inner while-loop) survives iff $e \in E'$ after the inner while-loop of algorithm \ref{alg:implrandmatch}.
\begin{lemma}
\label{lem:edgesurvives}
For every $e = \lbrace u,v \rbrace \in E$ with $deg_{E}(u) > 1$ or $\deg_{E}(v) > 1$ before the inner while-loop in algorithm \ref{alg:implrandmatch} the probability that $e$ survives is at least $\frac{1}{8 \cdot (deg_{E}(u) + \deg_{E}(v) - 2)}$.
\end{lemma}
\begin{proof}
%We modify the algorithm a little bit to make the analysis easier: In each iteration of the inner loop we delete each edge in $E'$ with probability $1/2$ (regardless of the degree of the incident nodes) and add all edges $(u,v)$ with $deg(u) = deg(v) = 1$ to a set $M'$. We do this until the graph, \ie $E'$, is empty and add $M'$ to $M$. In other words, instead of keeping the isolated edges in the graph we collect them in the set $M'$ which is added to $M$ after the loop. Thus, this modification does not change the set of edges we add to the matching.

Let $e = \lbrace u, v \rbrace \in E$ be an edge before the inner while-loop and $R_e$ be the number of rounds until edge $e$ is deleted. The random bits in each iteration are $3$-wise independent and the iterations themselves are completely independent. Thus, the variables $R_e$ are also $3$-wise independent. Denote by $N(e) = \lbrace e' \in E \mid e \cap e' \neq \emptyset \rbrace$ the neighborhood of $e$, \ie all edges incident to $u$ or $v$. Then we have
$ \Prob{}{e \text{ survives}} = \Prob{}{R_e \text{ is unique maximum in } \lbrace R_{e'} \mid e' \in N(e) \rbrace}$.
It is easy to see that $\Prob{}{R_e = i} = \left( \frac{1}{2} \right)^i$ for $i \geq 1$. Let $e' \in N(e)$ and $e' \neq e$ and $z \geq 1$ be fixed. Since the $R_e$ are $3$-wise independent, we have
$ \Prob{}{R_{e'} \geq z \mid R_e = z} = \Prob{}{R_{e'} \geq z} = \sum\limits_{i=z}^\infty \left( \frac{1}{2} \right)^i = \left( \frac{1}{2} \right)^{z-1}.$
Therefore, the probability that there is an edge $e' \in N(e) \setminus e$ with $R_{e'} \geq z$ is at most $\frac{\vert N(e) \vert - 1}{2^{z-1}}$, \ie $R_e$ is unique maximum with probability at least $1-\frac{\vert N(e) \vert - 1}{2^{z-1}}$. This is greater than $0$ for $z \geq \log(\vert N(e) \vert-1)+2$. Finally, we have
$\Prob{}{R_e \text{ is unique maximum}} \geq \left(\frac{1}{2}\right)^{\log(\vert N(e) \vert-1)+2} \cdot \left(1-\frac{\vert N(e) \vert - 1}{2^{\log(\vert N(e) \vert-1)+1}}\right) \geq \frac{1}{8 \cdot (deg_{E}(u) + \deg_{E}(v) - 2)}$
%$$ \begin{array}{rcl}
%\Prob{}{R_e \text{ is unique maximum}} & \geq & \left(\dfrac{1}{2}\right)^{\log(\vert N(e) \vert-1)+2} \cdot \left(1-\dfrac{\vert N(e) \vert - 1}{2^{\log(\vert N(e) \vert-1)+1}}\right)\\
%& = & \dfrac{1}{4 \cdot (\vert N(e) \vert-1)} \cdot \dfrac{1}{2}\\
%& = & \dfrac{1}{8 \cdot (deg_{E}(u) + \deg_{E}(v) - 2)}.
%\end{array}$$
%\Prob{}{R_e \text{ is unique maximum}} & \geq & \sum\limits_{z=\log(\vert N(e) \vert-1)+2}^\infty \left(\dfrac{1}{2}\right)^z \cdot \left(1-\dfrac{\vert N(e) \vert - 1}{2^{z-1}}\right)\\
%& \geq & \left(\dfrac{1}{2}\right)^{\log(\vert N(e) \vert-1)+2} \cdot \left(1-\dfrac{\vert N(e) \vert - 1}{2^{\log(\vert N(e) \vert-1)+1}}\right)\\
\end{proof}

%Algorithm \ref{alg:implrandmatch} shows the implicit version of Algorithm \ref{alg:randmaximalmatching}. 
The number of deleted edges for a matching edge $(u,v)$ that is added to the matching is $deg(u)+deg(v)-2$ if we do not count the matching edge itself. Thus, the expected number of deleted edges is $\Omega(\vert E \vert)$ at the end of the outer loop. This gives us the final result.
\begin{theorem}
\label{th:algorithm_operations}
Let $G = (V,E)$ be a graph with $N$ nodes. All functions used in algorithm \ref{alg:implrandmatch} on the input $\chi_E$ depend on at most $3 \log N$ variables. The expected number of operations is $O(\log^3 N)$.
\end{theorem}
\begin{proof}
Each iteration of the inner-loop needs $O(\log N)$ operations. 
Since we halve the number of edges in expectation in each iteration of this loop, the expected number of iterations is $O(\log N)$. The edges surviving the inner loop are those that are added to the matching. After adding a set of edges to the matching, all edges that are incident to a matched node are deleted from the graph in the outer loop. The number deleted edges for a matching edge $(u,v)$ that is added to the matching is $deg(u)+deg(v)-2$ if we do not count the matching edge itself. Thus, by Lemma 3, the expected number of edges deleted in this step is at least
$$ \sum_{e = \lbrace u, v \rbrace \in E} (deg_E(u) + deg_E(v)-2)\cdot\dfrac{1}{8 \cdot (deg_{E}(u) + \deg_{E}(v) - 2)} = \vert E \vert / 8. $$
This implies that the expected number of iterations of the outer-loop is also bounded above by $O(\log N)$.
\end{proof}

\subsubsection*{Application to the Maximal Independent Set Problem}
With a similar idea we are able to design a distributed MIS algorithm: Each node $v$ draws a random bit until this bit is $0$. Let $r_v$ be the number of bits drawn by node $v$. We send $r_v$ to all neighbors and include node $v$ to the independent set iff $r_v$ is a local minimum. The expected number of bits for each channel is $1$. A similar analysis as before show that we have an maximal independent set after $O(\log N)$ steps in expectation and the overall expected number of bits per channel is $O(\log N)$.

\subsection*{Experimental Results.} All algorithms are implemented in
C++ using the BDD framework CUDD 2.5.0\footnote{\url{http://vlsi.colorado.edu/~fabio/CUDD/}} by F. Somenzi and were compiled with Visual Studio 2013 in the default $32$-bit release configuration. All source files, scripts and random seeds will be publicly available\footnote{\url{http://ls2-www.cs.uni-dortmund.de/~gille/}}. The experiments were performed on a computer with a 2.5 GHz Intel Core i7 processor and 8 GB main memory running Windows 8.1. The runtime is measured by used processor time in seconds and the space usage of the implicit algorithm is given by the maximum SBDD size which came up during the computation, where an SBDD is a collection of OBDDs which can share nodes. Note that the maximum
SBDD size is independent of the used computer system. For our results, we took
the mean value over 50 runs on the same graph. Due to the small variance of these values, we only show the mean in the diagrams/tables.  
We omit the algorithm by Bollig and Pr\"{o}ger \cite{BolligP12} because the memory limitation was exceeded on every instance presented here. 

We choose three types of input instances: First, we used our construction from section \ref{sec:construction} as an input distribution in the following way: If the $1$ sink is chosen with probability $p$ as a successor of nodes in layer $L_{n-1}$ the expected size of $\vert f^{-1}(x) \vert$ is $p \cdot 2^n$. For a fixed $N = 2^{17}$, we used $p$ as a density parameter for our input graph and want to analyze how the density influences the running time of the algorithms. Second, we run the algorithms on some bipartite graphs from a real advertisement application within Google\footnote{Graph data files can be found at \url{http://www.columbia.edu/~cs2035/bpdata/}} \cite{NegPSSS09}. The motivation was to check whether the randomized algorithm is competitive or even better on instances where the algorithm by Hachtel and Somenzi (HS) \cite{HachtelS97} is running very well. Third, we use non-bipartite graphs from the university of Florida sparse matrix collection \cite{UFS}. Since HS is designed for bipartite graphs, a preprocessing step computing a bipartition of these graphs are needed to compute a maximal matching (see, \eg \cite{BolligP12}) while our algorithm also works on general graphs.

In the experiments we use the following implementation of our algorithm denoted by RM. In order to minimize the running time for computation of the set of nodes with two or more incident edges, we sparsify the graph at the beginning of the outer while loop by deleting each edge with probability $1/2$ and repeating this $D$ times. Initially, we set $D = \log \vert E \vert$ and decrease $D$ by $1$ at the end of the outer loop. Asymptotically, the running time does not change since after $O(\log N)$ iterations, \ie $D = 0$, it does exactly the same as original algorithm. Initial experiments showed that this is superior to the original algorithm.

\begin{figure}[H]
\includegraphics[width=0.49\textwidth]{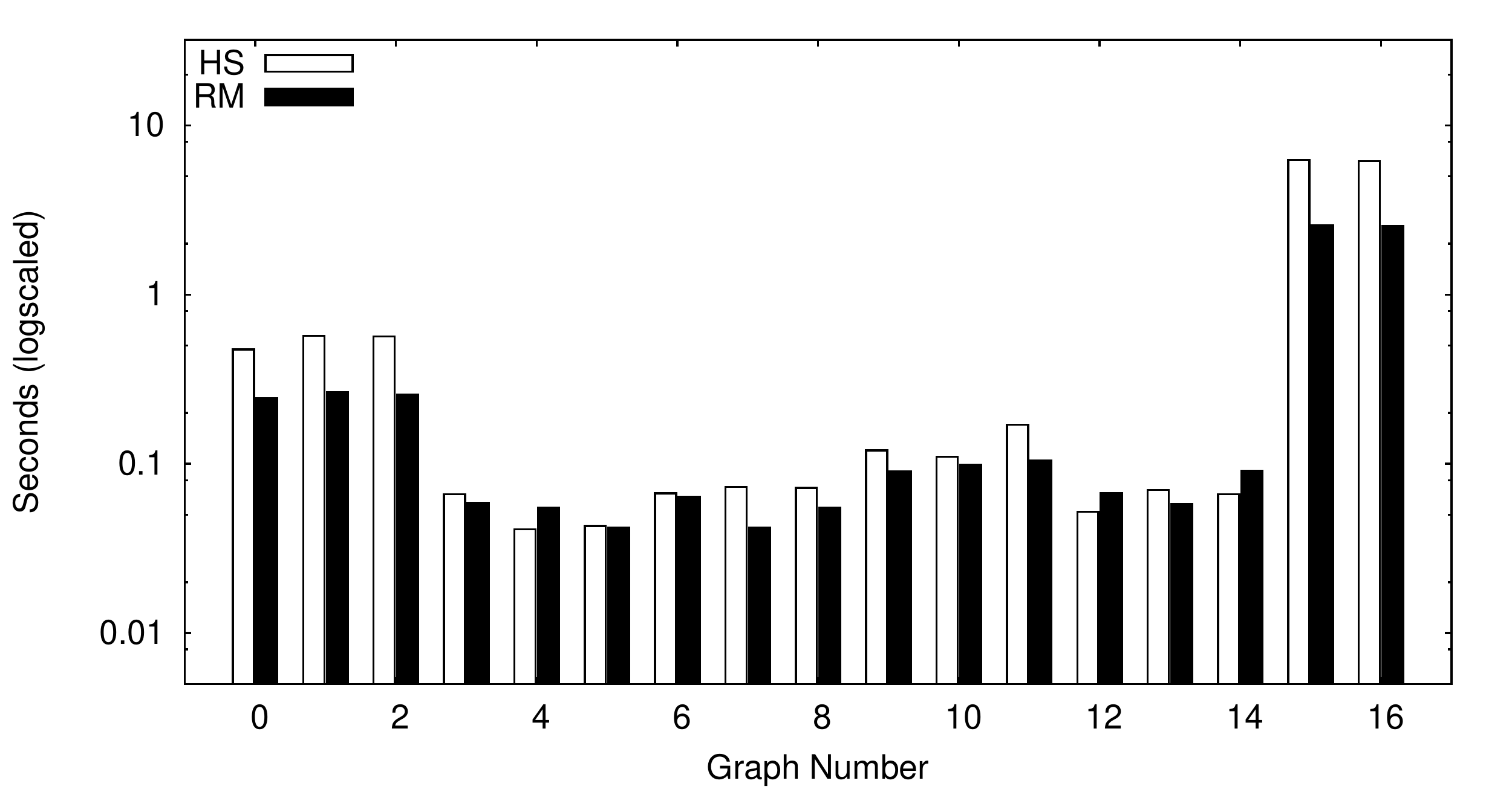}
\includegraphics[width=0.49\textwidth]{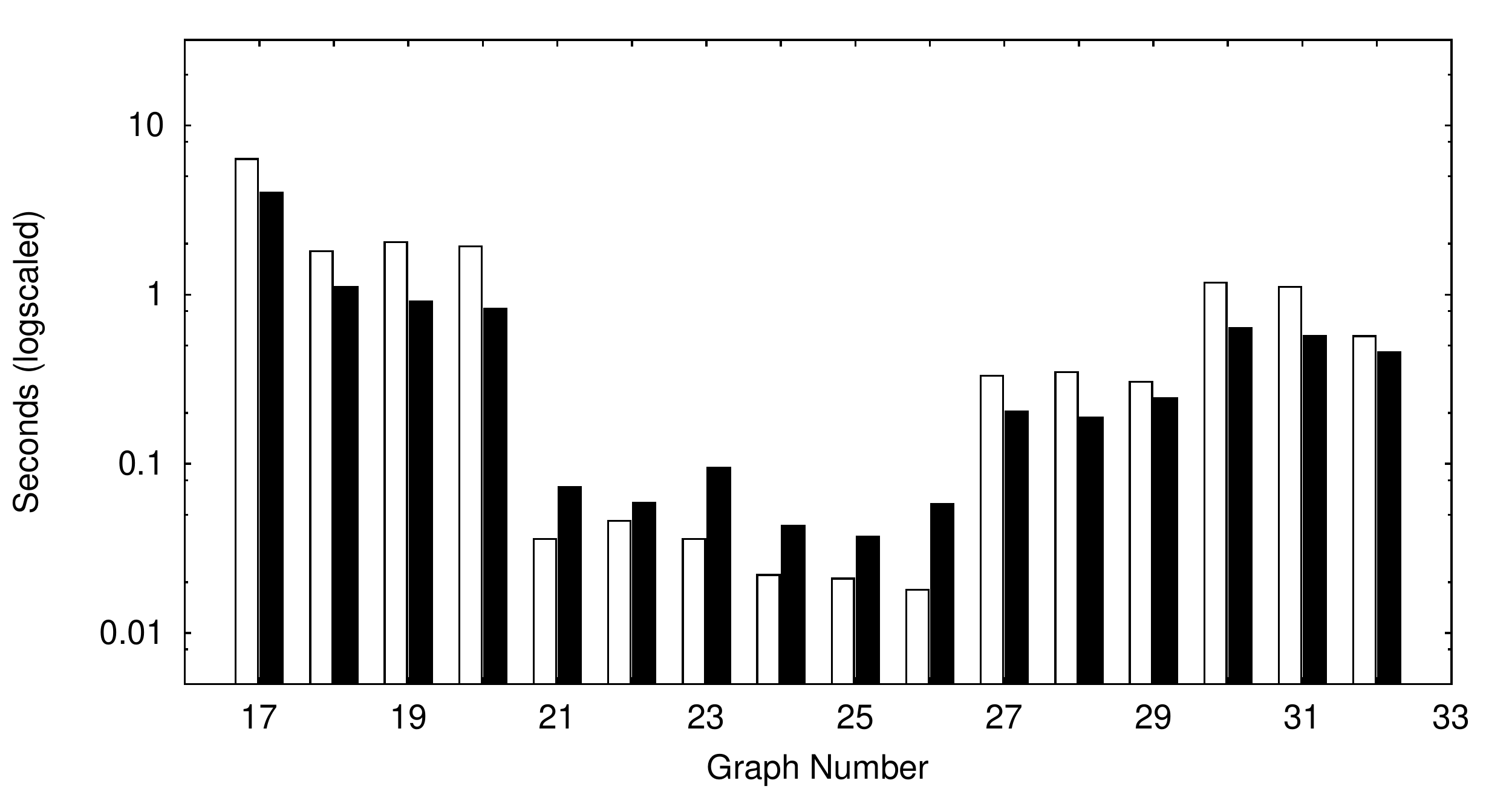}
\caption{Running times of HS and RM on the real world instances.}\label{fig:real}
\end{figure}

\begin{table}
\centering
\small 
\begin{tabular}{|c|c|c|c|c|}
\hline \textbf{Instance} & \textbf{Nodes}  & \textbf{Edges}  & \textbf{Time (sec)} & \textbf{Space (SBDD size)}  \\ 
\hline 333SP & 3712815  & 22217266  & 1140.54 & 66968594  \\ 
\hline adaptive & 6815744  & 27248640  & 403.82 & 22767094\\ 
\hline as-Skitter & 1696415  & 22190596  & 337.53 & 32020282 \\ 
\hline hollywood-2009 & 1139905 & 113891327 & 418.36 & 62253086 \\
\hline roadNet-CA & 1971281 & 5533214  & 136.18 & 13177668 \\ 
\hline roadNet-PA & 1090920 & 3083796  & 75.26 & 7633318 \\ 
\hline roadNet-TX & 1393383 & 3843320 & 92.62 & 9125438 \\ 
\hline 
\end{tabular} 
\caption{Running time and space usage of RM on the graphs from \cite{UFS}}\label{tab:ufs}
\end{table}

On the random instances the running time and space usage of RM was more or less unaffected by the density of the graph while HS was very slow for small values of $p$ and gets faster with increasing density. For $p \leq 0.2$ RM was much faster than HS (see Fig.\,\ref{fig:random_runtime} \ref{fig:random_space}). In Fig.\,\ref{fig:real} we see that on the bipartite real world instances RM is similar to HS if the running time is negligibly small but on the largest instances (number 15 to 20) RM is much faster. the graphs from \cite{UFS} were intentionally chosen to show the potential of RM and indeed do so: It was not possible to run HS on these graphs due to memory limitations whereas RM computed a matching in reasonable time and space (see Table\,\ref{tab:ufs}). Both graphs from and \cite{UFS} have very small density and the experiments on the random graphs seem to support the hypothesis that RM is a better choice than HS for such graphs.

\begin{figure}[H]
\includegraphics[width=\textwidth]{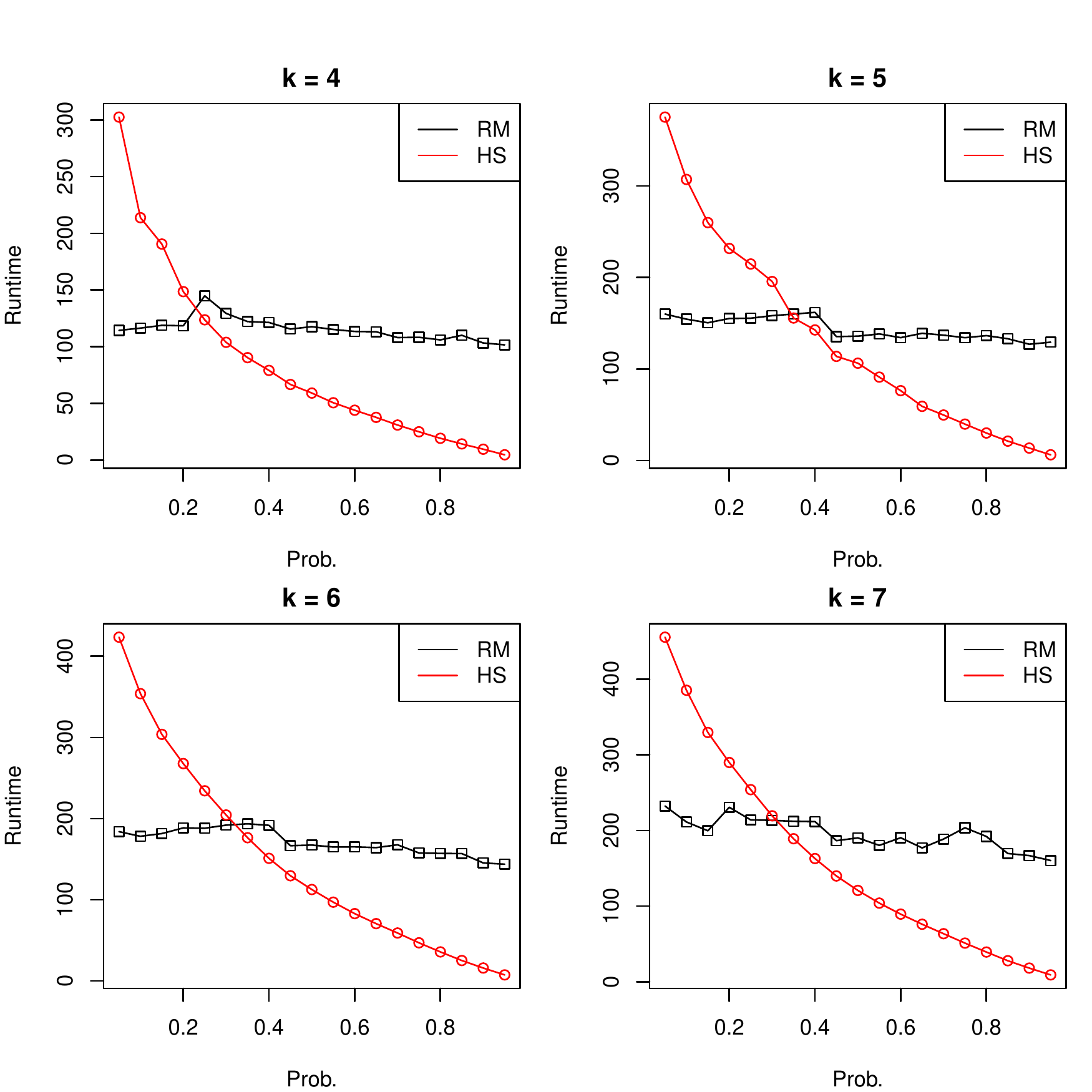}
\caption{Running times of HS and RM on the random instances.}\label{fig:random_runtime}
\end{figure}

\begin{figure}[H]
\includegraphics[width=\textwidth]{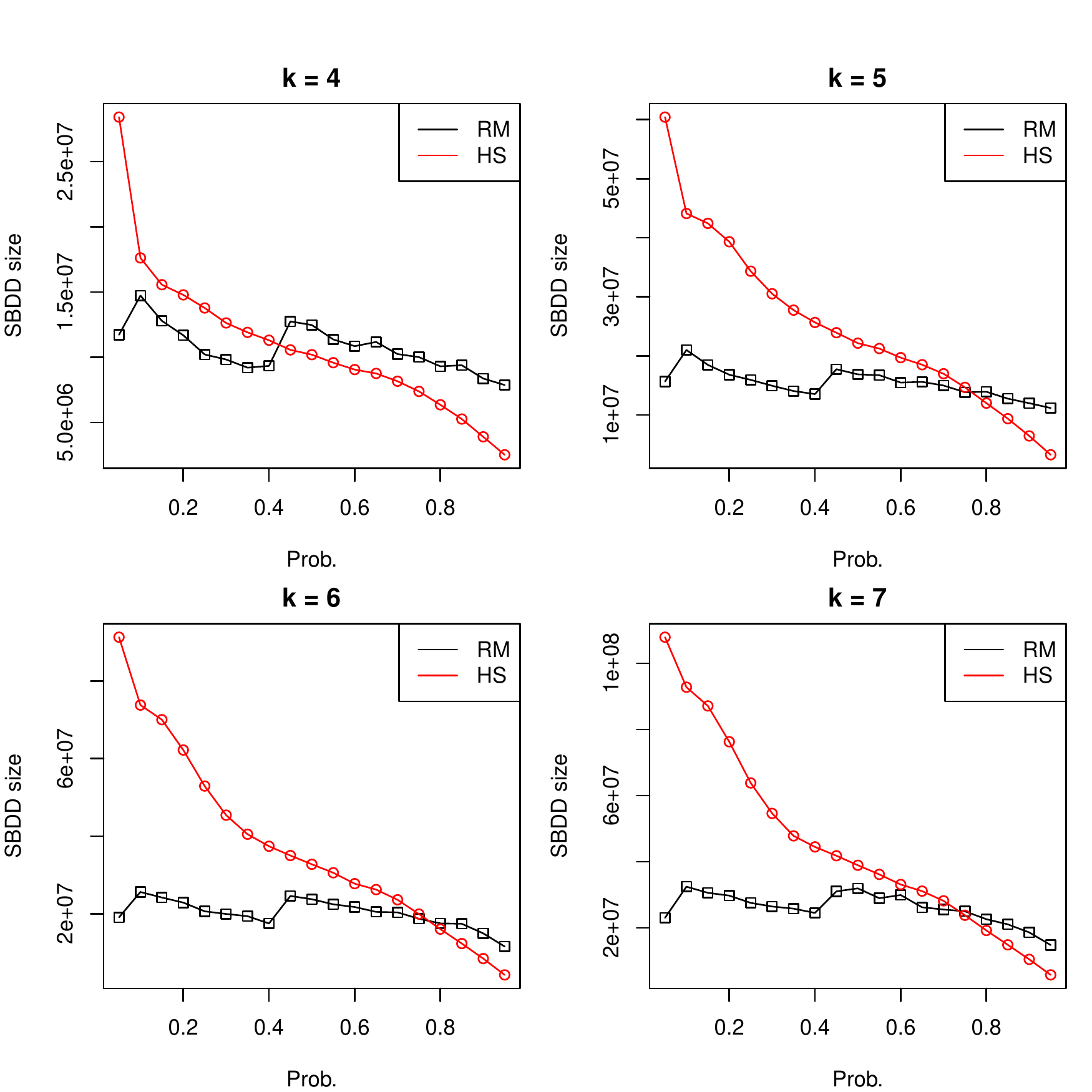}
\caption{Space usage of HS and RM on the random instances.}\label{fig:random_space}
\end{figure}

\subsubsection*{Further Applications.} 
Extending our matching algorithm to the more general $f$-matching, where each node $v$ is allowed to have at most $f(v)$ incident matching edges, is an interesting question. Designing other randomized implicit algorithms, \eg for minimum spanning tree, where random sampling of subgraphs are necessary, seems straightforward and initial experiments showed that this could lead to faster algorithms than the known deterministic ones.

%\section{Conclusion}

%\input{conclusion}

\subsubsection*{Acknowledgements} I would like to thank Beate Bollig, Melanie Schmidt and Chris Schwiegelshohn for the valuable discussions and for their comments on the presentation of the paper.

\bibliographystyle{acm} 
\bibliography{literatur}

\begin{thebibliography}{10}

\bibitem{Alon86}
{\sc Alon, N., Babai, L., and Itai, A.}
\newblock A fast and simple randomized parallel algorithm for the maximal
  independent set problem.
\newblock {\em J. Algorithms 7}, 4 (1986), 567 -- 583.

\bibitem{AlonGHP92}
{\sc Alon, N., Goldreich, O., H{\aa}stad, J., and Peralta, R.}
\newblock Simple construction of almost k-wise independent random variables.
\newblock {\em Random Struct. Alg. 3}, 3 (1992), 289--304.

\bibitem{AlonMS99}
{\sc Alon, N., Matias, Y., and Szegedy, M.}
\newblock The space complexity of approximating the frequency moments.
\newblock {\em J. Comp. and System Sc. 58}, 1 (1999), 137 -- 147.

\bibitem{AwerbuchGLP89}
{\sc Awerbuch, B., Goldberg, A.~V., Luby, M., and Plotkin, S.~A.}
\newblock Network decomposition and locality in distributed computation.
\newblock In {\em FOCS\/} (1989), pp.~364--369.

\bibitem{BloemGS06}
{\sc Bloem, R., Gabow, H.~N., and Somenzi, F.}
\newblock An algorithm for strongly connected component analysis in {\it
  n}log{\it n} symbolic steps.
\newblock {\em Formal Meth. in System Design 28}, 1 (2006), 37--56.

\bibitem{Bollig12}
{\sc Bollig, B.}
\newblock On symbolic {OBDD}-based algorithms for the minimum spanning tree
  problem.
\newblock {\em Theor. Comput. Sci. 447\/} (2012), 2--12.

\bibitem{BolligC13}
{\sc Bollig, B., and Capelle, M.}
\newblock Priority functions for the approximation of the metric {TSP}.
\newblock {\em Inf. Proc. Letters 113}, 14-16 (2013), 584--591.

\bibitem{BolligGP12}
{\sc Bollig, B., Gill{\'e}, M., and Pr{\"o}ger, T.}
\newblock Implicit computation of maximum bipartite matchings by sublinear
  functional operations.
\newblock In {\em TAMC\/} (2012), pp.~473--486.

\bibitem{BolligLW96}
{\sc Bollig, B., L{\"o}bbing, M., and Wegener, I.}
\newblock On the effect of local changes in the variable ordering of ordered
  decision diagrams.
\newblock {\em Inf. Proc. Letters 59}, 5 (1996), 233--239.

\bibitem{BolligP12}
{\sc Bollig, B., and Pr{\"o}ger, T.}
\newblock An efficient implicit {OBDD}-based algorithm for maximal matchings.
\newblock In {\em LATA\/} (2012), pp.~143--154.

\bibitem{BolligW14}
{\sc Bollig, B., and Wegener, I.}
\newblock personal communication, 2014.

\bibitem{Bryant86}
{\sc Bryant, R.~E.}
\newblock Graph-based algorithms for boolean function manipulation.
\newblock {\em IEEE Transactions on Computers 35}, 8 (1986), 677--691.

\bibitem{BurchCMDH92}
{\sc Burch, J.~R., Clarke, E.~M., McMillan, K.~L., Dill, D.~L., and Hwang,
  L.~J.}
\newblock Symbolic model checking: $10^{20}$ states and beyond.
\newblock {\em Inf. and Comp. 98}, 2 (1992), 142--170.

\bibitem{ChorG89}
{\sc Chor, B., and Goldreich, O.}
\newblock On the power of two-point based sampling.
\newblock {\em J. Complexity 5}, 1 (1989), 96--106.

\bibitem{Coudert95}
{\sc Coudert, O.}
\newblock Doing two-level logic minimization 100 times faster.
\newblock In {\em SODA\/} (1995), pp.~112--121.

\bibitem{UFS}
{\sc Davis, T.~A., and Hu, Y.}
\newblock {The University of Florida Sparse Matrix Collection}.
\newblock {\em {ACM} Trans. on Math. Soft. 38}, 1 (Nov. 2011), 1:1--1:25.

\bibitem{GentiliniPP03}
{\sc Gentilini, R., Piazza, C., and Policriti, A.}
\newblock Computing strongly connected components in a linear number of
  symbolic steps.
\newblock In {\em SODA\/} (2003), pp.~573--582.

\bibitem{GentiliniPP08}
{\sc Gentilini, R., Piazza, C., and Policriti, A.}
\newblock Symbolic graphs: Linear solutions to connectivity related problems.
\newblock {\em Algorithmica 50}, 1 (2008), 120--158.

\bibitem{Gille13}
{\sc Gill{\'e}, M.}
\newblock {OBDD}-based representation of interval graphs.
\newblock In {\em WG}, vol.~8165 of {\em LNCS}. Springer Berlin Heidelberg,
  2013, pp.~286--297.

\bibitem{HachtelS97}
{\sc Hachtel, G.~D., and Somenzi, F.}
\newblock A symbolic algorithms for maximum flow in 0-1 networks.
\newblock {\em F. Meth. in Sys. Design 10}, 2/3 (1997), 207--219.

\bibitem{HojatiTKB93}
{\sc Hojati, R., Touati, H., Kurshan, R.~P., and Brayton, R.~K.}
\newblock Efficient $\omega$-regular language containment.
\newblock In {\em Comp. Aided Verification}, vol.~663 of {\em LNCS}. Springer,
  1993, pp.~396--409.

\bibitem{IsraelI86}
{\sc Israeli, A., and Itai, A.}
\newblock A fast and simple randomized parallel algorithm for maximal matching.
\newblock {\em Inf. Process. Lett. 22}, 2 (1986), 77--80.

\bibitem{Jukna88}
{\sc Jukna, S.}
\newblock Entropy of contact circuits and lower bounds on their complexity.
\newblock {\em Theor. Comput. Sci. 57\/} (1988), 113--129.

\bibitem{Kabanets03}
{\sc Kabanets, V.}
\newblock Almost k-wise independence and hard boolean functions.
\newblock {\em Theor. Comput. Sci. 297}, 1-3 (2003), 281--295.

\bibitem{LaiPV94}
{\sc Lai, Y., Pedram, M., and Vrudhula, S. B.~K.}
\newblock {EVBDD}-based algorithms for integer linear programming, spectral
  transformation, and function decomposition.
\newblock {\em IEEE Trans. on CAD of Int. Circuits and Systems 13}, 8 (1994),
  959--975.

\bibitem{Linial92}
{\sc Linial, N.}
\newblock Locality in distributed graph algorithms.
\newblock {\em SIAM J. Comput. 21}, 1 (1992), 193--201.

\bibitem{Luby86}
{\sc Luby, M.}
\newblock A simple parallel algorithm for the maximal independent set problem.
\newblock {\em SIAM Journal on Computing 15}, 4 (1986), 1036--1053.

\bibitem{Masek76}
{\sc Masek, W.}
\newblock A fast algorithm for the string editing problem and decision graph
  complexity.
\newblock Master's thesis, MIT, 1976.

\bibitem{MeerR09}
{\sc Meer, K., and Rautenbach, D.}
\newblock On the {OBDD} size for graphs of bounded tree- and clique-width.
\newblock {\em Discrete Mathematics 309}, 4 (2009), 843--851.

\bibitem{MetivierRSZ11}
{\sc M{\'e}tivier, Y., Robson, J.~M., Saheb-Djahromi, N., and Zemmari, A.}
\newblock An optimal bit complexity randomized distributed {MIS} algorithm.
\newblock {\em Distributed Computing 23}, 5-6 (2011), 331--340.

\bibitem{NaorN93}
{\sc Naor, J., and Naor, M.}
\newblock Small-bias probability spaces: Efficient constructions and
  applications.
\newblock {\em SIAM J. Comput. 22}, 4 (1993), 838--856.

\bibitem{NegPSSS09}
{\sc Negruseri, C.~S., Pasoi, M.~B., Stanley, B., Stein, C., and Strat, C.~G.}
\newblock Solving maximum flow problems on real world bipartite graphs.
\newblock In {\em ALENEX\/} (2009), pp.~14--28.

\bibitem{NuWo09}
{\sc Nunkesser, R., and Woelfel, P.}
\newblock Representation of graphs by {OBDDs}.
\newblock {\em Discrete Applied Mathematics 157}, 2 (2009), 247--261.

\bibitem{Savicky95}
{\sc Savick{\'y}, P.}
\newblock Improved boolean formulas for the ramsey graphs.
\newblock {\em Random Struct. Algorithms 6}, 4 (1995), 407--416.

\bibitem{Sawitzki04}
{\sc Sawitzki, D.}
\newblock Implicit flow maximization by iterative squaring.
\newblock In {\em SOFSEM\/} (2004), pp.~301--313.

\bibitem{SawitzkiSOFSEM06}
{\sc Sawitzki, D.}
\newblock The complexity of problems on implicitly represented inputs.
\newblock In {\em SOFSEM\/} (2006), pp.~471--482.

\bibitem{SawitzkiLATIN06}
{\sc Sawitzki, D.}
\newblock Exponential lower bounds on the space complexity of {OBDD}-based
  graph algorithms.
\newblock In {\em LATIN\/} (2006), pp.~781--792.

\bibitem{Sawitzki07}
{\sc Sawitzki, D.}
\newblock Implicit simulation of {FNC} algorithms.
\newblock {\em Electronic Colloquium on Computational Complexity (ECCC) 14},
  028 (2007).

\bibitem{SielingW93}
{\sc Sieling, D., and Wegener, I.}
\newblock {NC}-algorithms for operations on binary decision diagrams.
\newblock {\em Parallel Processing Letters 3\/} (1993), 3--12.

\bibitem{Weg94a}
{\sc Wegener, I.}
\newblock The size of reduced {OBDDs} and optimal read-once branching programs
  for almost all boolean functions.
\newblock {\em IEEE Trans. on Comp. 43}, 11 (1994), 1262--1269.

\bibitem{Wegener00}
{\sc Wegener, I.}
\newblock {\em Branching programs and binary decision diagrams}.
\newblock SIAM Monographs on Discrete Mathematics and Applications, 2000.

\bibitem{Woe2006}
{\sc Woelfel, P.}
\newblock Symbolic topological sorting with {OBDDs}.
\newblock {\em J. Disc. Alg. 4\/} (2006), 51--71.

\end{thebibliography}

\newpage
\appendix
\subsection*{Proof of Theorem \ref{th:lower_bound_general}}
\begin{claim}
\label{cl:entropy}
Let $\varepsilon \leq 1/2$. Then $3\varepsilon+H(\varepsilon) \leq 6 \sqrt{\varepsilon}$.
\end{claim}
\begin{proof}
Recall that $H(x) = -x \log(x)-(1-x)\log(1-x)$. Using $\dfrac{1}{1-x} \leq 1+2x$ for $0 \leq x \leq 1/2$ and $\log(1+x) \leq x$ for $x > -1$ we have
\begin{eqnarray*}
3\varepsilon+H(\varepsilon) & = & 3\varepsilon + \varepsilon  \log(1/\varepsilon)+(1-\varepsilon) \cdot \log(1/(1-\varepsilon))\\
& \leq & 3\varepsilon + \varepsilon \sqrt{1/\varepsilon} + \log(1/(1-\varepsilon))\\
& \leq & 3\varepsilon+\sqrt{\varepsilon} + 2\varepsilon\\
& \leq & 6 \sqrt{\varepsilon}. 
\end{eqnarray*}
\end{proof}

\subsection*{Experiments}

%\begin{figure}[h]
%\includegraphics[width=\textwidth]{images/random_runtime_mean.pdf}
%\caption{Running times of HS and RM on the random instances.}\label{fig:random_runtime}
%\end{figure}

\begin{figure}[h]
\includegraphics[width=\textwidth]{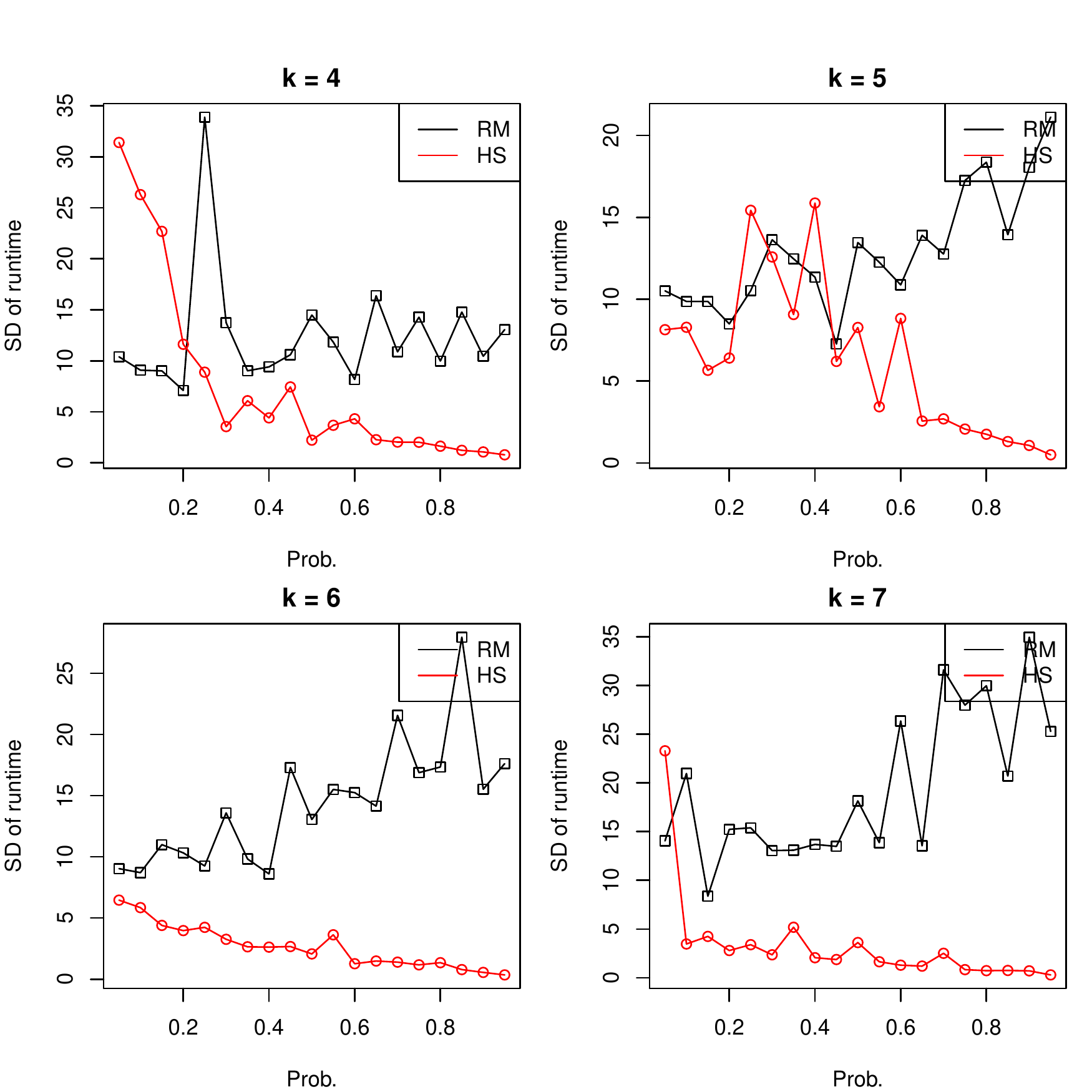}
\caption{Standard deviations of the running times.}\label{fig:random_runtime_sd}
\end{figure}

%\begin{figure}[h]
%\includegraphics[width=\textwidth]{images/random_space.pdf}
%\caption{Space usage of HS and RM on the random instances.}\label{fig:random_space}
%\end{figure}

\begin{figure}[h]
\includegraphics[width=\textwidth]{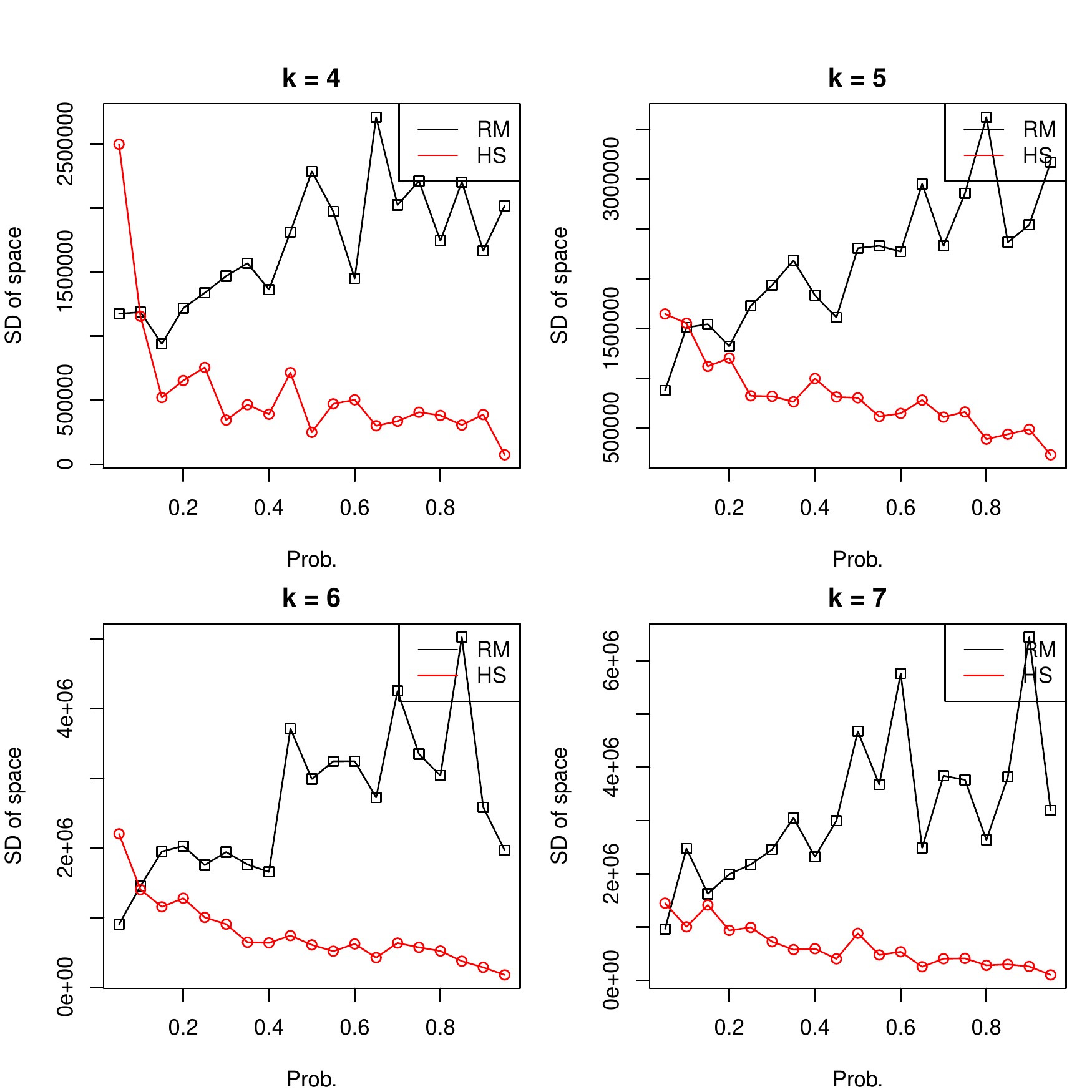}
\caption{Standard deviation of the space usage.}\label{fig:random_space_sd}
\end{figure}

\begin{table}[h]
\begin{center}
\begin{tabular}{|c|c|c|c|c|}
\hline
Number & Running Time (RM) & Running Time (HS) & SBDD Size (RM) & SBDD Size (HS)  \\ \hline
$0$ & \textbf{0.243} & 0.475 & \textbf{567210} & 1346996 \\ \hline
$1$ & \textbf{0.264} & 0.571 & \textbf{555968} & 1394008 \\ \hline
$2$ & \textbf{0.256} & 0.567 & \textbf{553924} & 1394008 \\ \hline
$3$ & \textbf{0.059} & 0.066 & \textbf{153300} & 220752 \\ \hline
$4$ & 0.055 & \textbf{0.041} & \textbf{161476} & 194180 \\ \hline
$5$ & \textbf{0.042} & 0.043 & \textbf{153300} & 194180 \\ \hline
$6$ & \textbf{0.064} & 0.067 & \textbf{196224} & 252434 \\ \hline
$7$ & \textbf{0.042} & 0.073 & \textbf{163520} & 279006 \\ \hline
$8$ & \textbf{0.055} & 0.072 & \textbf{169652} & 279006 \\ \hline
$9$ & \textbf{0.09} & 0.12 & \textbf{240170} & 368942 \\ \hline
$10$ & \textbf{0.099} & 0.11 & \textbf{237104} & 368942 \\ \hline
$11$ & \textbf{0.105} & 0.17 & \textbf{245280} & 368942 \\ \hline
$12$ & 0.067 & \textbf{0.052} & \textbf{236082} & 245280 \\ \hline
$13$ & \textbf{0.058} & 0.07 & \textbf{242214} & 310688 \\ \hline
$14$ & 0.091 & \textbf{0.066} & \textbf{284116} & 328062 \\ \hline
$15$ & \textbf{2.565} & 6.259 & \textbf{3115056} & 7887796 \\ \hline
\end{tabular}
\end{center}
\caption{Running times and space usage of RM and HS on real-world instances from \textup{\cite{NegPSSS09}}.}
\end{table}

\begin{table}[h]
\begin{center}
\begin{tabular}{|c|c|c|c|c|}
\hline
Number & Running Time (RM) & Running Time (HS) & SBDD Size (RM) & SBDD Size (HS)  \\ \hline
$16$ & \textbf{2.545} & 6.167 & \textbf{3115056} & 7874510 \\ \hline
$17$ & \textbf{4.002} & 6.329 & \textbf{3115056} & 7874510 \\ \hline
$18$ & \textbf{1.112} & 1.81 & \textbf{2053198} & 2320962 \\ \hline
$19$ & \textbf{0.913} & 2.043 & \textbf{2035824} & 2485504 \\ \hline
$20$ & \textbf{0.828} & 1.931 & \textbf{2035824} & 2485504 \\ \hline
$21$ & 0.073 & \textbf{0.036} & \textbf{182938} & 231994 \\ \hline
$22$ & 0.059 & \textbf{0.046} & \textbf{163520} & 240170 \\ \hline
$23$ & 0.095 & \textbf{0.036} & \textbf{162498} & 240170 \\ \hline
$24$ & 0.043 & \textbf{0.022} & 134904 & 134904 \\ \hline
$25$ & 0.037 & \textbf{0.021} & 135926 & 135926 \\ \hline
$26$ & 0.058 & \textbf{0.018} & 169652 & \textbf{135926} \\ \hline
$27$ & \textbf{0.203} & 0.331 & \textbf{346458} & 731752 \\ \hline
$28$ & \textbf{0.188} & 0.348 & \textbf{317842} & 677586 \\ \hline
$29$ & \textbf{0.244} & 0.305 & \textbf{319886} & 677586 \\ \hline
$30$ & \textbf{0.632} & 1.176 & \textbf{1314292} & 2100210 \\ \hline
$31$ & \textbf{0.568} & 1.114 & \textbf{1280566} & 2104298 \\ \hline
$32$ & \textbf{0.458} & 0.568 & \textbf{950460} & 1410360 \\ \hline
\end{tabular}
\end{center}
\caption{Running times and space usage of RM and HS on real-world instances from \textup{\cite{NegPSSS09}}.}
\end{table}

\begin{table}[h]
\begin{center}
\begin{tabular}{|c|c|c|c|}
\hline
Number & $\vert U \vert$ & $\vert W \vert$ & Edges\\ \hline
$0$ & $136$ & $18872$ & $222951$   \\ \hline
$1$ & $137$ & $18872$ & $222951$  \\ \hline
$2$ & $137$ & $18888$ & $222951$  \\ \hline
$3$ & $40$ & $7086$ & $33609$  \\ \hline
$4$ & $41$ & $7086$ & $33609$  \\ \hline
$5$ & $41$ & $7093$ & $33609$  \\ \hline
$6$ & $125$ & $7107$ & $33609$  \\ \hline
$7$ & $86$ & $7117$ & $33609$  \\ \hline
$8$ & $86$ & $7127$ & $33609$  \\ \hline
$9$ & $289$ & $16653$ & $33051$  \\ \hline
$10$ & $290$ & $16846$ & $33051$  \\ \hline
$11$ & $290$ & $16904$ & $33051$  \\ \hline
$12$ & $50$ & $13360$ & $50040$  \\ \hline
$13$ & $51$ & $16264$ & $56577$  \\ \hline
$14$ & $51$ & $21016$ & $56577$  \\ \hline
$15$ & $164$ & $288826$ & $2523313$  \\ \hline
\end{tabular}
\begin{tabular}{|c|c|c|c|c|}
\hline
Number & $\vert U \vert$ & $\vert W \vert$ & Edges \\ \hline
$16$ & $165$ & $288826$ & $2523313$  \\ \hline
$17$ & $165$ & $288858$ & $2523313$  \\ \hline
$18$ & $196$ & $89030$ & $1080027$  \\ \hline
$19$ & $197$ & $89044$ & $1080041$  \\ \hline
$20$ & $197$ & $89044$ & $1080041$  \\ \hline
$21$ & $40$ & $28489$ & $43629$  \\ \hline
$22$ & $41$ & $28944$ & $43629$  \\ \hline
$23$ & $41$ & $28944$ & $43629$  \\ \hline
$24$ & $35$ & $11361$ & $22279$  \\ \hline
$25$ & $36$ & $11588$ & $22279$  \\ \hline
$26$ & $36$ & $11695$ & $22279$  \\ \hline
$27$ & $934$ & $8752$ & $42711$  \\ \hline
$28$ & $935$ & $8896$ & $42711$  \\ \hline
$29$ & $935$ & $9028$ & $42711$  \\ \hline
$30$ & $125$ & $55058$ & $844598$   \\ \hline
$31$ & $126$ & $56858$ & $844598$   \\ \hline
$31$ & $126$ & $57926$ & $477356$   \\ \hline
\end{tabular}
\end{center}
\caption{Properties of the real-world instances from \textup{\cite{NegPSSS09}}.}
\label{tab:realworld}
\end{table}

\end{document}